\documentclass[11pt,a4paper]{article}

\usepackage[utf8]{inputenc}
\usepackage[T1]{fontenc}
\usepackage[a4paper,top=2.5cm,bottom=2.5cm,left=2.5cm,right=2.5cm]{geometry}
\usepackage{amsmath,amssymb,amsthm,amsfonts}
\usepackage{mathtools}
\usepackage{enumitem}
\usepackage{hyperref}
\usepackage{xcolor}
\usepackage{tikz}
\usetikzlibrary{calc}

\usepackage{microtype}

\theoremstyle{plain}
\newtheorem{theorem}{Theorem}[section]
\newtheorem{lemma}[theorem]{Lemma}
\newtheorem{proposition}[theorem]{Proposition}
\newtheorem{corollary}[theorem]{Corollary}

\theoremstyle{definition}
\newtheorem{definition}[theorem]{Definition}
\newtheorem{example}[theorem]{Example}
\newtheorem{remark}[theorem]{Remark}

\newcommand{\R}{\mathbb{R}}
\newcommand{\N}{\mathbb{N}}
\newcommand{\GN}{\mathcal{G}^N}
\newcommand{\RN}{\mathbb{R}^N}
\newcommand{\Sh}{\mathrm{Sh}}
\newcommand{\Shi}{\mathrm{Sh}_i}
\newcommand{\Bz}{\mathrm{Bz}}
\newcommand{\Bzi}{\mathrm{Bz}_i}
\newcommand{\ED}{\mathrm{ED}}
\newcommand{\EDi}{\mathrm{ED}_i}
\newcommand{\ESD}{\mathrm{ESD}}
\newcommand{\ESDi}{\mathrm{ESD}_i}
\newcommand{\So}{\mathrm{So}}

\newcommand{\Feps}{F^{\varepsilon}}
\newcommand{\Fepsi}{F^{\varepsilon}_i}
\newcommand{\mV}{V}
\newcommand{\mL}{L}

\newcommand{\eps}{\varepsilon}

\newcommand{\Var}{\mathrm{Var}}
\newcommand{\E}{\mathbb{E}}

\title{The Geometry of Cooperative Game Solutions:\\ Stratified Egalitarian Shapley Values}

\author{Frank M.\,V.\,Feys} 
\date{\today}

\begin{document}

\maketitle

\begin{abstract}
\noindent The space $\mL$ of linear value maps on a finite-player cooperative game $\GN$ is finite-dimensional, and admits a canonical inner product induced by the Harsanyi-dividend decomposition of $\GN$. 
We show that this inner product is intrinsic: the same value arises from any orthonormal basis of $\GN$ with respect to the Harsanyi inner product, not merely from the unanimity basis used to define it. 
Within this geometry, the subspace $\mL^{\mathrm{ESL}}$ of efficient, symmetric, linear value maps admits a clean structure theorem. 
The induced orthogonal stratification of $\mL$ by coalition size yields a canonical linear isomorphism $\mL^{\mathrm{ESL}} \cong \R^{n-1}$, under which every efficient symmetric linear value map decomposes uniquely into $n - 1$ \emph{stratified epsilons}, one per coalition size. 
The classical egalitarian Shapley family of Joosten (1996) is, in fact, precisely the diagonal slice of this $\R^{n-1}$. 
The orthogonal projection of any $\Psi \in \mL^{\mathrm{ESL}}$ onto this diagonal yields an optimal parameter $\eps^*(\Psi)$ which equals the weighted mean of the stratified epsilons under an explicit probability distribution $\{w_a\}$ over coalition sizes, and the associated goodness-of-fit measure $R^2(\Psi)$ equals one minus the relative weighted variance of those epsilons. 
The framework is therefore a literal regression-statistics analogue of the coefficient of determination, with the underlying probability space being $(\{1, \ldots, n - 1\}, w)$. 
The framework yields explicit closed-form computations for natural solutions. 
At $n = 4$ it produces a clean three-way classification of the standard alternatives to the Shapley value: the Banzhaf value is nearly orthogonal to the egalitarian Shapley axis ($R^2 \approx 1\%$); the equal-surplus-division value is moderately aligned ($R^2 \approx 38\%$); the solidarity value is almost entirely aligned ($R^2 \approx 99.6\%$). 
Asymptotic analysis sharpens the contrast: $R^2(\ESD) \to 1$ exponentially fast, $R^2(\So) \to 1$ via vanishing relative weighted variance of the stratified epsilons, and $R^2(\Bz) \to 1/2$, the last reflecting a structural identity between the efficiency defect and the egalitarian-Shapley deviation of the Banzhaf value at every coalition size. 
Each conclusion is a one-line corollary of the stratified structure theorem. 
\end{abstract}

\section{Introduction}\label{sec-intro}

\subsection*{Motivation}

Cooperative game theory has produced a wide range of single-valued solution concepts that allocate the worth of a coalition $v(N)$ among its members. 
The most widely studied is the Shapley value, which is characterized by Shapley's classical four-axiom theorem: efficiency, symmetry, linearity, and the dummy property. 
Yet many alternatives have been proposed and studied, each motivated by a different fairness intuition: the Banzhaf value (averaging marginal contributions over coalitions), the equal-division solution (pure egalitarian sharing), the equal-surplus-division value (egalitarian sharing of the joint surplus above individual worths), the solidarity value of Nowak and Radzik (averaging marginal contributions within each coalition), and the egalitarian Shapley values of Joosten (convex combinations of the Shapley value and the equal-division solution).

Note that each such alternative has its own axiomatization, and the literature contains many comparison results phrased axiomatically. 
What is largely missing, however, is a \emph{quantitative} way to compare them. 
How close is the Banzhaf value to the Shapley value, in some intrinsic sense, exactly?
 Is the equal-surplus-division value structurally similar to the equal-division solution, or genuinely distinct from it? 
 Are the various ``egalitarian-flavored'' alternatives to Shapley really expressing one common idea, or several independent ones? More fundamentally: is there a canonical structure on the space of solutions in which such questions admit principled answers?

These questions are difficult to address, since a priori there is no canonical metric on the space of solutions.
A natural starting point is to think of solutions as elements of a function space and equip it with an inner product, but the obvious construction (summing $\langle \Psi(v), \Phi(v) \rangle_{\R^N}$ over a basis of games) depends, a priori, on the choice of basis.
 We show that the right inner product is intrinsic, that the space of efficient symmetric linear solutions has a clean canonical parametrization in terms of it, and that this parametrization gives every quantitative question about ``closeness to the Shapley value'' a precise geometric meaning.

\subsection*{Contributions}

In this paper we propose a geometric framework that addresses these questions.
 Our starting point is the observation that the space $\mL$ of \emph{linear} value maps
 is  finite-dimensional once the player set $N$ is fixed: $\dim(\mL) = (2^n - 1) \cdot n$.
   Virtually every solution studied in cooperative game theory is linear in $v$, so this restriction loses no important examples.

The space $\mL$ in fact admits a canonical inner product, which is induced by the Harsanyi-dividend decomposition of $\GN$ (Definition~\ref{def-IP-on-L}).
 Our first technical result (Theorem~\ref{thm-basis-indep}) shows that this inner product is intrinsic:
  the same value arises from any orthonormal basis of $\GN$ with respect to the Harsanyi inner product, not merely from the unanimity basis used to define it.

The structural core of the paper (Section~\ref{sec-spectral}) is as follows.
The Harsanyi inner product induces an orthogonal stratification of $\mL$ by coalition size: $\mL = \bigoplus_{a=1}^{n} \mL_a$, where $\mL_a$ consists of value maps supported on $a$-coalitions.
 Restricting to the symmetric and efficient subspace, each stratum becomes one-dimensional and is generated by the stratum-$a$ component of $\ED - \Sh$. 
We obtain a canonical linear isomorphism
\[
\mL^{\mathrm{ESL}} \xrightarrow{\;\sim\;} \R^{n-1}, \qquad \Psi \longmapsto \big(\eps_1(\Psi), \ldots, \eps_{n-1}(\Psi)\big),
\]
where $\Psi^{(a)} = (1 - \eps_a) \Sh^{(a)} + \eps_a \ED^{(a)}$ stratum by stratum (Theorem~\ref{thm-stratified-flat}).
 Every efficient symmetric linear value map is therefore a \emph{stratified egalitarian Shapley value},
  with one parameter per coalition size; and the classical egalitarian Shapley family of Joosten is, in fact, the one-parameter diagonal slice of this $\R^{n-1}$.

Within this picture, the orthogonal projection onto the egalitarian Shapley line decomposes rather cleanly.
 The optimal parameter $\eps^*(\Psi)$ is the weighted mean of the stratified epsilons under an explicit probability distribution $\{w_a\}$ over coalition sizes (Theorem~\ref{thm-stratified-proj}); the goodness-of-fit measure $R^2(\Psi)$ is one minus the relative weighted variance of those epsilons (Theorem~\ref{thm-R2-variance}).
  The framework is thus a literal regression-statistics analogue of the coefficient of determination, with $(\{1, \ldots, n-1\}, w)$ as the underlying probability space and the diagonal $\eps_1 = \ldots = \eps_{n-1}$ as the model class.
  The framework also extends to nonefficient symmetric linear value maps via an additional efficiency-defect coordinate at each stratum (Section~\ref{subsec-Bz-stratified}).

Applied to the standard alternatives to the Shapley value,
 the framework yields explicit closed-form computations and a clean three-way classification.
  At $n = 4$:
\begin{itemize}[leftmargin=2em]
\item the \emph{Banzhaf value} is nearly orthogonal to the egalitarian Shapley axis ($R^2 \approx 1\%$);
\item the \emph{equal-surplus-division value} is moderately aligned ($R^2 \approx 38\%$);
\item the \emph{solidarity value} is almost entirely aligned ($R^2 \approx 99.6\%$).
\end{itemize}
Asymptotic analysis as $n \to \infty$ sharpens this contrast (Theorems~\ref{thm-asym-ESD} and~\ref{thm-asym-Bz}, 
together with Equation~\eqref{eq-So-asymp} and Corollary~\ref{cor-Bz-asym-stratified}).
 For $\ESD$ the deviation from the egalitarian Shapley line is concentrated entirely on the singleton stratum, whose weight $w_1$ vanishes asymptotically, hence $R^2(\ESD) \to 1$.
  For the solidarity value, $\eps_a(\So) \to 1$ pointwise in $a$ as $n \to \infty$ and the relative weighted variance vanishes, hence $R^2(\So) \to 1$.
   For the Banzhaf value, the efficiency defect satisfies the structural identity $\delta_a(\Bz) = -\eps_a(\Bz)/n$ at every stratum, which forces the squared norms of the efficient and uniform components to match asymptotically, hence $R^2(\Bz) \to 1/2$.
   The Banzhaf value is therefore not just numerically different from the egalitarian Shapley family, but \emph{geometrically} so, in a way that no growth in the player set can erase.

We see this as a quantitative articulation of an intuition that working game theorists have long held informally; namely,
 the Banzhaf value ``feels different'' from the Shapley value in a way that the equal-surplus-division and solidarity values do not.
  Our framework provides the first formal vocabulary for that intuition, and grounds it in a structure theorem rather than a numerical coincidence.

We additionally recover, in our notation, the characterization of the egalitarian Shapley family of Joosten (1996) via an $\varepsilon$-Dummy axiom that directly relaxes the Dummy axiom of Shapley's classical theorem (Section~\ref{sec-eps-shapley}),
 and we extend the projection framework to richer multi-parameter affine flats through the Shapley value, with explicit closed-form coefficients via the normal equations (Section~\ref{sec-multiparam}).

\subsection*{Related Work}

The egalitarian Shapley family $\{\Feps  \mid \eps \in [0, 1]\}$ originates in Joosten~\cite{joosten1996} and has been characterized axiomatically several times since:
 by van den Brink, Funaki, and Ju~\cite{vandenbrink2013} via efficiency, linearity, anonymity, and weak monotonicity, and by Casajus and Huettner~\cite{casajus2013} via efficiency, additivity, desirability, and a null-player-in-a-productive-environment property.
  The $\eps$-Dummy axiom we use (Section~\ref{sec-eps-shapley}) is, after the reparameterization $\alpha = 1 - \eps$, the $\alpha$-egalitarian null player property attributed to Joosten in~\cite[footnote~2]{casajus2013};
   we include the corresponding four-axiom characterization for completeness, as the direct analogue of Shapley's classical theorem.

The Shapley value itself admits multiple axiomatizations beyond Shapley's original four-axiom theorem,
 notably Young's~\cite{young1985} strong-monotonicity characterization and the potential-function approach of Hart and Mas-Colell~\cite{hartmascolell1989}; we work with the original four-axiom formulation throughout.
  Existing comparisons of $\Sh$, $\ED$, $\ESD$, $\Bz$, and $\So$ in the literature have been primarily axiomatic or based on specific structural properties such as null-player robustness~\cite{kamijokongo2012,casajus2014};
   the present paper instead provides a single inner-product geometry in which all five solutions occupy precise quantitative positions.

What is new here is not the family $\{\Feps\}$ itself, but the inner-product geometry on $\mL$ within which the family becomes a canonical affine line, and the quantitative comparison of standard solutions to that line via projection.

\subsection*{Outline}

Section~\ref{sec-coop-games} fixes notation and recalls the Shapley value. 
Section~\ref{sec-eps-shapley} introduces the $\eps$-Dummy axiom and recovers the egalitarian Shapley family axiomatically. 
Sections~\ref{sec-value-space} and~\ref{sec-basis-indep} construct the inner product on $\mL$ and prove its basis-independence. 
Section~\ref{sec-best-approx} computes the optimal $\eps^*$ for the Banzhaf and equal-surplus-division values. 
Section~\ref{sec-pythagorean} sets up the Pythagorean decomposition, defines $R^2(\Psi)$, 
computes it for our three targets at $n = 4$, and presents the three-way classification. 
Section~\ref{sec-spectral} establishes the structural core of the paper:
 the orthogonal stratification of $\mL$ by coalition size, the canonical isomorphism $\mL^{\mathrm{ESL}} \cong \R^{n-1}$ via stratified epsilons, and the closed-form expression of $R^2$ as one minus the relative weighted variance of those epsilons. 
 Section~\ref{sec-asymptotic} establishes the asymptotic results $R^2(\ESD) \to 1$ and $R^2(\Bz) \to 1/2$ and uses the stratified picture to obtain $R^2(\So) \to 1$.
  Section~\ref{sec-multiparam} extends the framework to multi-parameter projection. 
  Section~\ref{sec-discussion} discusses open directions.
   Appendix~\ref{app-So-formula} derives the closed form for the stratified epsilons of the solidarity value.

\section{Cooperative Games and the Shapley Value}\label{sec-coop-games}

Let $N \subseteq \N$ be a finite set of players with $|N| = n \geq 2$.
 A \emph{cooperative game with transferable utility} (TU-game) is a pair $(N, v)$, where $v \colon 2^N \to \R$ is a \emph{characteristic function} satisfying $v(\emptyset) = 0$.
  Any subset $C \subseteq N$ is called a \emph{coalition}, and $v(C) \in \R$ is its \emph{worth}. 
  We take $N$ to be fixed and refer to the game by $v$ alone.

The set of all TU-games on $N$, denoted $\GN$, carries a vector-space structure under the pointwise operations $(v + w)(C) = v(C) + w(C)$ and $(\alpha v)(C) = \alpha \cdot v(C)$.
 This space has dimension $2^n - 1$. 
 For each nonempty $A \subseteq N$, the \emph{unanimity game} $u_A \in \GN$ is defined by
\[
u_A(C) = \begin{cases} 1 & \text{if } A \subseteq C, \\ 0 & \text{otherwise.} \end{cases}
\]
The collection $\{u_A \mid \emptyset \neq A \subseteq N\}$ forms a basis of $\GN$. 
Each game $v \in \GN$ has a unique decomposition
\[
v = \sum_{\emptyset \neq A \subseteq N} h_v(A) \cdot u_A,
\]
in which the coefficients $h_v(A) \in \R$ are called the \emph{Harsanyi dividends}~\cite{harsanyi1959} of $v$.

A (single-valued) \emph{solution} for TU-games is a map $\Psi \colon \GN \to \RN$ that assigns a payoff vector $\Psi(v) \in \RN$ to each game $v \in \GN$. 
We write $\Psi = (\Psi_i)_{i \in N}$.

The \emph{Shapley value} $\Sh \colon \GN \to \RN$, introduced by Shapley~\cite{shapley1953value},
 is given by
\[
\Shi(v) := \frac{1}{n!} \sum_{\pi \in \Pi(N)} m_i^\pi(v),
\]
for each $i \in N$, where $\Pi(N)$ is the set of permutations of $N$, $S_\pi(i) = \{j \in N \mid \pi(j) < \pi(i)\}$ denotes the set of strict predecessors of $i$ under $\pi$, and $m_i^\pi(v) := v(S_\pi(i) \cup \{i\}) - v(S_\pi(i))$ is the marginal contribution of $i$.
 The Shapley value of a unanimity game is $\Shi(u_A) = 1/|A|$ if $i \in A$, and $\Shi(u_A) = 0$ otherwise.

The \emph{equal-division solution} $\ED \colon \GN \to \RN$ is given by $\EDi(v) := v(N)/n$ for all $i \in N$.

We recall the four classical axioms for a solution $\Psi$:
\begin{description}[leftmargin=2.4cm,labelwidth=2.0cm,labelindent=0pt]
\item[(Eff)] $\sum_{i \in N} \Psi_i(v) = v(N)$ for all $v \in \GN$.
\item[(Sym)] If $i, j \in N$ are $v$-interchangeable, i.e., $v(S \cup \{i\}) = v(S \cup \{j\})$ 
for every $S \subseteq N$ containing neither, then $\Psi_i(v) = \Psi_j(v)$.
\item[(Lin)] $\Psi$ is linear: $\Psi(\alpha v + \beta w) = \alpha \Psi(v) + \beta \Psi(w)$ for all $v, w \in \GN$ and $\alpha, \beta \in \R$.
\item[(Dum)] If $i \in N$ is a \emph{dummy player} (i.e., $v(S \cup \{i\}) = v(S)$ for every $S \subseteq N$), then $\Psi_i(v) = 0$.
\end{description}

The classical theorem of Shapley~\cite{shapley1953value} states that there is exactly one solution satisfying (Eff), (Sym), (Lin) and (Dum), namely the Shapley value $\Sh$.

\section{The $\varepsilon$-Dummy Axiom and the $\varepsilon$-Shapley Value}\label{sec-eps-shapley}

The Dummy axiom (Dum) demands that a player who contributes nothing also receives nothing.
 While this is a natural marginalist principle, it is in some contexts arguably too strong: 
 a member of a society who is unable to contribute may nonetheless be entitled to a share of the collective surplus.
 In order to accommodate this, we propose a one-parameter relaxation of (Dum):

\begin{description}[leftmargin=2.4cm,labelwidth=2.0cm,labelindent=0pt]
\item[($\eps$-Dum)] If $i \in N$ is a dummy player, then $\Psi_i(v) = \eps \cdot v(N)/n$.
\end{description}

When $\eps = 0$,  this reduces to the classical Dummy axiom. 
When $\eps = 1$, a dummy player receives exactly the average payoff $v(N)/n$.
 For $\eps \in (0,1)$, the dummy receives a positive fraction of that average, 
 providing the ``guaranteed minimum'' we sought.

For each $\eps \in \R$, define $\Feps \colon \GN \to \RN$ by
\[
\Fepsi(v) := (1 - \eps) \Shi(v) + \eps \cdot \frac{v(N)}{n}, \qquad i \in N.
\]
We call $\Feps$ the \emph{$\eps$-Shapley value}.
 Equivalently, $\Feps = (1-\eps)\Sh + \eps \cdot \ED$, exhibiting $\Feps$ as a convex combination of the marginalist Shapley value and the egalitarian equal-division solution.

\begin{remark}\label{rem-joosten}
The family $\{\Feps \mid \eps \in [0,1]\}$ coincides with the so-called \emph{egalitarian Shapley values} introduced by Joosten~\cite{joosten1996} and studied extensively, with several axiomatic characterizations: 
first by Joosten himself~\cite{joosten1996}, who used parameterized standardness for two-player games and consistency; 
then by van den Brink, Funaki and Ju~\cite{vandenbrink2013}, who used efficiency, linearity, anonymity and weak monotonicity;
and by Casajus and Huettner~\cite{casajus2013} via efficiency, additivity, desirability and the null-player-in-a-productive-environment property. 
The axiom \emph{($\eps$-Dum)} corresponds, after the reparameterization $\alpha = 1 - \eps$, to the \emph{$\alpha$-egalitarian null player property} attributed to Joosten in~\cite[footnote~2]{casajus2013}. 
The characterization below is the corresponding direct analogue of Shapley's original four-axiom theorem and is included for completeness; 
the substantive contribution of this paper begins in Section~\ref{sec-value-space}.
\end{remark}

\begin{theorem}\label{thm-eps-shapley}
Let $\eps \in \R$. 
There is a unique solution $\Psi \colon \GN \to \RN$ satisfying \emph{(Eff)}, \emph{(Sym)}, \emph{(Lin)} and \emph{($\eps$-Dum)}, namely $\Psi = \Feps$.
\end{theorem}

\begin{proof}
We first verify that $\Feps$ satisfies all four axioms.

\smallskip
\textsc{(Eff).} 
Using the efficiency of the Shapley value, $\sum_{i \in N} \Shi(v) = v(N)$:
\[
\sum_{i \in N} \Fepsi(v) = (1-\eps)\sum_{i \in N} \Shi(v) + \eps \sum_{i \in N} \frac{v(N)}{n} = (1-\eps)v(N) + \eps \cdot v(N) = v(N).
\]

\smallskip
\textsc{(Sym).} 
If $i, j$  are $v$-interchangeable, then by the symmetry of the Shapley value, 
we get $\Shi(v) = \Sh_j(v)$, and consequently $\Fepsi(v) = F^\eps_j(v)$.

\smallskip
\textsc{(Lin).} 
Both $v \mapsto \Shi(v)$ and $v \mapsto v(N)/n$  are linear in $v$,
 hence so is their convex combination.

\smallskip
\textsc{($\eps$-Dum).} 
If $i$ is a dummy  player, then $\Shi(v) = 0$, so $\Fepsi(v) = \eps \cdot v(N)/n$.

\medskip

We now show uniqueness. 
Suppose $\Psi$ satisfies all four axioms. 
We treat the cases $\eps \neq 1$ and $\eps = 1$ separately.

\smallskip
\emph{Case $\eps \neq 1$.} 
Define $G \colon \GN \to \RN$ by
\[
G_i(v) := \frac{1}{1-\eps} \Psi_i(v) - \frac{\eps}{1-\eps} \cdot \frac{v(N)}{n}, \qquad i \in N.
\]
We show  that $G$  satisfies (Eff), (Sym), (Lin) and (Dum), so that by Shapley's theorem $G = \Sh$, 
whence rearranging gives $\Psi = \Feps$.
For (Eff), using the efficiency of $\Psi$:
\[
\sum_{i \in N} G_i(v) = \frac{1}{1-\eps} v(N) - \frac{\eps}{1-\eps} \cdot n \cdot \frac{v(N)}{n} = \frac{v(N) - \eps \cdot v(N)}{1-\eps} = v(N).
\]
For (Sym), this is immediate from the  symmetry of $\Psi$ and the constant term $v(N)/n$ being identical for all $i$.
For (Lin), $G$ is a linear combination of linear maps, hence linear.
For (Dum), if $i$ is a dummy, then $\Psi_i(v) = \eps \cdot v(N)/n$ by ($\eps$-Dum), 
so
\[
G_i(v) = \frac{1}{1-\eps} \cdot \eps \cdot \frac{v(N)}{n} - \frac{\eps}{1-\eps} \cdot \frac{v(N)}{n} = 0.
\]

\smallskip
\emph{Case $\eps = 1$.} 
Suppose $\Psi$ satisfies $(1$-Dum$)$ and the other three axioms. 
Define $T \colon \GN \to \RN$ by $T_i(v) := \Psi_i(v) - v(N)/n$.
 Then $T$ is linear (sum of linear maps), and we shall show that $T \equiv 0$.
Since $\{u_A \mid \emptyset \neq A \subseteq N\}$ is a basis of $\GN$ and $T$ is linear, 
it suffices to prove $T(u_A) = 0$ for every nonempty $A \subseteq N$.
Fix such an $A$. 
In the game $u_A$, every player $i \in N \setminus A$   is a dummy. 
By (1-Dum) applied to $\Psi$, 
we have $\Psi_i(u_A) = u_A(N)/n = 1/n$, so $T_i(u_A) = 0$ for $i \in N \setminus A$.
For any $i, j \in A$, players $i$ and $j$ are $u_A$-interchangeable
 (since adding either to a coalition $S \not\supseteq A \setminus \{i,j\}$ 
 has the same effect on $u_A$). 
 By (Sym), $T_i(u_A) = T_j(u_A)$.
Finally, summing $T$ over $N$ and using (Eff) for $\Psi$:
\[
\sum_{i \in N} T_i(u_A) = \sum_{i \in N} \Psi_i(u_A) - n \cdot \frac{u_A(N)}{n} = u_A(N) - u_A(N) = 0.
\]
Since $T_i(u_A) = 0$ for $i \notin A$ and $T_i(u_A)$ is constant for $i \in A$ with the total summing to zero, 
we conclude $T_i(u_A) = 0$ for all $i \in N$. 
Thus, $T \equiv 0$ and $\Psi = F^1$.
\end{proof}

\section{The Space of Value Maps and an Inner Product}\label{sec-value-space}

In this section, we lift the perspective from individual solutions to the entire collection of value maps.

\begin{definition}\label{def-V}
The \emph{space of value maps} is
\[
\mV = \{\Psi \mid \Psi \colon \GN \to \RN\},
\]
equipped with the pointwise vector-space operations, given by 
 $(\Psi + \Phi)(v) = \Psi(v) + \Phi(v)$ and $(\alpha \Psi)(v) = \alpha \cdot \Psi(v)$ 
 for $\Psi, \Phi \in \mV$, $v \in \GN$ 
 and $\alpha \in \R$.
\end{definition}

The space $\mV$ is infinite-dimensional. 
However, in cooperative game theory, virtually every solution under study is, in fact, linear in $v$: 
this is the case for the Shapley value, the Banzhaf value, the equal-division solution, 
the equal-surplus-division solution, the solidarity value, weighted Shapley values, and so on. 
We therefore restrict attention to the subspace of linear value maps:
\[
\mL = \{\Psi \in \mV \mid \Psi \text{ is linear in } v\}.
\]
By basic linear algebra, $\dim(\mL) = \dim(\GN) \cdot \dim(\RN) = (2^n - 1) \cdot n$, 
so $\mL$ is finite-dimensional.

The egalitarian Shapley line $\{\Feps \mid \eps \in \R\}$ lies entirely within $\mL$, 
since both $\Sh$ and $\ED$ are linear.

\subsection{The Harsanyi Inner Product}

To define a canonical inner product on $\mL$, we equip $\GN$ itself with a canonical inner product induced by the Harsanyi-dividend decomposition.

\begin{definition}\label{def-harsanyi-IP}
The \emph{Harsanyi inner product} on $\GN$ is
\[
\langle v, w \rangle_H = \sum_{\emptyset \neq A \subseteq N} h_v(A) \cdot h_w(A),
\]
where $h_v(A)$ and $h_w(A)$ are the Harsanyi dividends of $v$ and $w$, respectively.
\end{definition}

By construction, the unanimity basis $\{u_A \mid \emptyset \neq A \subseteq N\}$ is orthonormal with respect to $\langle \cdot, \cdot \rangle_H$. 
We call a basis $B$ of $\GN$ \emph{$H$-orthonormal} if it is orthonormal with respect to $\langle \cdot, \cdot \rangle_H$.

\subsection{The Induced Inner Product on \texorpdfstring{$\mL$}{L}}

We define an inner product on $\mL$ by summing the standard inner product on $\RN$ over the unanimity basis:

\begin{definition}\label{def-IP-on-L}
For $\Psi, \Phi \in \mL$, let
\[
\langle \Psi, \Phi \rangle_{\mL} = \sum_{\emptyset \neq A \subseteq N} \langle \Psi(u_A), \Phi(u_A) \rangle_{\RN},
\]
where $\langle \cdot, \cdot \rangle_{\RN}$ is the standard Euclidean inner product on $\RN$.
\end{definition}

\begin{lemma}\label{lem-IP-valid}
$\langle \cdot, \cdot \rangle_{\mL}$ is an inner product on $\mL$.
\end{lemma}

\begin{proof}
Bilinearity and symmetry follow directly from the corresponding properties of $\langle \cdot, \cdot \rangle_{\RN}$.
For positive-definiteness: $\langle \Psi, \Psi \rangle_{\mL} = \sum_A \|\Psi(u_A)\|^2 \geq 0$, with equality 
if and only if 
$\Psi(u_A) = 0$ for all $A$. 
Since $\Psi$ is linear and $\{u_A\}$ spans $\GN$, this forces $\Psi \equiv 0$.
\end{proof}

\section{Basis Independence}\label{sec-basis-indep}

The inner product introduced in Definition~\ref{def-IP-on-L} is built using the unanimity basis. 
A natural question is whether this dependence is essential, or whether the same inner product would arise from any other natural basis of $\GN$. 
The following theorem provides the answer to this question: 
the inner product is invariant under any change to a basis that is orthonormal with respect to the Harsanyi inner product on $\GN$. 
In this sense, $\langle \cdot, \cdot \rangle_{\mL}$ is a canonical structure on $\mL$, 
induced intrinsically by the Harsanyi inner product, rather than an artifact of basis choice.

In practice, every computation in the rest of this paper will use the unanimity basis, since it is the most convenient. 
The basis-independence theorem below should therefore be seen as a conceptual consistency check, rather than as a computational tool.

\begin{theorem}[Basis Independence]\label{thm-basis-indep}
Let $B = \{v_1, \ldots, v_d\}$, with $d = 2^n - 1$, be any $H$-orthonormal basis of $\GN$. 
For all $\Psi, \Phi \in \mL$,
\[
\sum_{k=1}^d \langle \Psi(v_k), \Phi(v_k) \rangle_{\RN} = \langle \Psi, \Phi \rangle_{\mL}.
\]
\end{theorem}

\begin{proof}
Express each $v_k$ in the unanimity basis:
 $v_k := \sum_{A} M_{kA} \cdot u_A$, where $A$ ranges over nonempty subsets of $N$. 
 The matrix $M = (M_{kA})$ is orthogonal. 
  Indeed, since both $B$ and the unanimity basis are $H$-orthonormal,
\[
\delta_{k\ell} = \langle v_k, v_\ell \rangle_H = \sum_{A,B} M_{kA} M_{\ell B} \langle u_A, u_B \rangle_H = \sum_A M_{kA} M_{\ell A},
\]
so $M M^T = I$,  hence  also $M^T M = I$, i.e., $\sum_k M_{kA} M_{kB} = \delta_{AB}$.
Since $\Psi, \Phi \in \mL$ are linear, $\Psi(v_k) = \sum_A M_{kA} \Psi(u_A)$ and $\Phi(v_k) = \sum_B M_{kB} \Phi(u_B)$. 
Therefore,
\begin{align*}
\sum_{k=1}^d \langle \Psi(v_k), \Phi(v_k) \rangle_{\RN}
 &= \sum_{k=1}^d \left\langle \sum_A M_{kA} \Psi(u_A), \sum_B M_{kB} \Phi(u_B) \right\rangle_{\RN} \\
&= \sum_{A,B} \left( \sum_{k=1}^d M_{kA} M_{kB} \right) \langle \Psi(u_A), \Phi(u_B) \rangle_{\RN} \\
&= \sum_{A,B} \delta_{AB} \langle \Psi(u_A), \Phi(u_B) \rangle_{\RN} \\
&= \sum_A \langle \Psi(u_A), \Phi(u_A) \rangle_{\RN} \\
&= \langle \Psi, \Phi \rangle_{\mL}. \qedhere
\end{align*}
  \end{proof}

\begin{remark}
The proof depends crucially on the linearity of $\Psi$ and $\Phi$:
 it is the linearity that allows the matrix $M$ to be pulled outside the value map. 
 For nonlinear value maps, the analogous statement fails, 
 and the inner product genuinely depends on the choice of basis.
\end{remark}

\section{Best Approximation by an $\varepsilon$-Shapley Value}\label{sec-best-approx}

In this section, we use the inner-product structure on $\mL$ to find, 
for any linear value map $\Psi \in \mL$, the $\eps$-Shapley value that is closest to $\Psi$.

\begin{definition}
Let $\Psi \in \mL$. 
The \emph{optimal $\eps$ for $\Psi$} is
\[
\eps^*(\Psi) = \arg\min_{\eps \in \R} \|\Feps - \Psi\|_{\mL},
\]
where $\|\cdot\|_{\mL}$ is the norm induced by $\langle \cdot, \cdot \rangle_{\mL}$.
\end{definition}

\begin{theorem}\label{thm-optimal-eps}
For every $\Psi \in \mL$, the optimal $\eps$ exists, is unique, and is given by
\[
\eps^*(\Psi) = \frac{\langle \ED - \Sh, \Psi - \Sh \rangle_{\mL}}{\langle \ED - \Sh, \ED - \Sh \rangle_{\mL}}.
\]
\end{theorem}

\begin{proof}
We can write $\Feps = \Sh + \eps(\ED - \Sh)$. 
Minimizing $\|\Sh + \eps(\ED - \Sh) - \Psi\|^2_{\mL}$ over $\eps$ is the projection of $\Psi - \Sh$ onto the line spanned by $\ED - \Sh$. 
Setting the derivative with respect to $\eps$ to zero,
\[
\frac{d}{d\eps} \|\Sh + \eps(\ED - \Sh) - \Psi\|^2_{\mL} = 2 \langle \ED - \Sh, \, \Sh + \eps(\ED - \Sh) - \Psi \rangle_{\mL} = 0,
\]
yields the stated formula. 
The denominator is nonzero, since $\ED \neq \Sh$ for $n \geq 2$.
\end{proof}

By Theorem~\ref{thm-basis-indep}, 
$\eps^*(\Psi)$ is well-defined and may be computed using any $H$-orthonormal basis. 
We compute it explicitly using the unanimity basis.

\subsection{Computing the Denominator}

For $|A| = a$, the vector $(\ED - \Sh)(u_A) \in \RN$ has components
\[
(\ED - \Sh)_i(u_A) = 
\begin{cases}
 \frac{1}{n} - \frac{1}{a} & \text{if } i \in A, \\ \frac{1}{n} & \text{if } i \notin A. 
 \end{cases}
\]
Therefore,
\begin{align*}
\|(\ED - \Sh)(u_A)\|^2 &= a \left(\frac{1}{n} - \frac{1}{a}\right)^2 + (n-a)\left(\frac{1}{n}\right)^2 \\
&= a \left( \frac{1}{n^2} - \frac{2}{na} + \frac{1}{a^2} \right) + \frac{n-a}{n^2} \\
&= \frac{a}{n^2} - \frac{2}{n} + \frac{1}{a} + \frac{n-a}{n^2} \\
&= \frac{1}{n} - \frac{2}{n} + \frac{1}{a} = \frac{1}{a} - \frac{1}{n}.
\end{align*}
Summing over all nonempty subsets,
  and using that there are $\binom{n}{a}$ subsets of size $a$,
\[
D_n := \langle \ED - \Sh, \ED - \Sh \rangle_{\mL} = \sum_{a=1}^n \binom{n}{a}\left(\frac{1}{a} - \frac{1}{n}\right) = H_n - \frac{2^n - 1}{n},
\]
where $H_n := \sum_{a=1}^n \binom{n}{a}/a$.

\begin{lemma}\label{lem-Dn-values}
For small $n$, the constant $D_n$ takes the values
\[
D_2 = 1, \qquad D_3 = \tfrac{5}{2}, \qquad D_4 = \tfrac{29}{6}, \qquad D_5 = \tfrac{103}{12}.
\]
\end{lemma}

\begin{proof}
Direct computation. 
For $n = 2$: $H_2 = 2 + 1/2 = 5/2$, so $D_2 = 5/2 - 3/2 = 1$. 
For $n = 3$: $H_3 = 3 + 3/2 + 1/3 = 29/6$, so $D_3 = 29/6 - 7/3 = 5/2$. 
The remaining values follow analogously.
\end{proof}

\subsection{Closed-Form Computations of $\eps^*$}

We now compute $\eps^*(\Psi)$ for several natural choices of $\Psi$.

\begin{proposition}\label{prop-trivial-cases}
$\eps^*(\Sh) = 0$ and $\eps^*(\ED) = 1$.
\end{proposition}

\begin{proof}
For $\Psi = \Sh$, the numerator is $\langle \ED - \Sh, 0 \rangle_{\mL} = 0$. 
For $\Psi = \ED$, the numerator equals the denominator. \qedhere
\end{proof}

This is the expected sanity check: $F^0 = \Sh$ projects to itself, 
and $F^1 = \ED$ projects to itself.

\subsubsection*{The Banzhaf Value}

The \emph{Banzhaf value}~\cite{banzhaf1965weighted} $\Bz$ is given by
\[
\Bzi(v) := \frac{1}{2^{n-1}} \sum_{C \subseteq N \setminus \{i\}} (v(C \cup \{i\}) - v(C)).
\]
On unanimity games, 
we have
\[
\Bzi(u_A) = 
\begin{cases}
 \frac{1}{2^{a-1}} & \text{if } i \in A, \\ 0 & \text{if } i \notin A,
 \end{cases}
\qquad a = |A|.
\]

\begin{proposition}\label{prop-eps-banzhaf}
\[
\eps^*(\Bz) = 1 + \frac{2 - 2 \cdot (3/2)^{n-1}}{D_n}.
\]
\end{proposition}

\begin{proof}
We compute the numerator $\langle \ED - \Sh, \Bz - \Sh \rangle_{\mL}$. 
 For fixed $A$ with $|A| = a$,
\begin{align*}
(\Bz - \Sh)_i(u_A) &= 
\begin{cases}
 \frac{1}{2^{a-1}} - \frac{1}{a} & \text{if } i \in A, \\ 0 & \text{if } i \notin A.  
 \end{cases}
\end{align*}
Therefore,
\[
\langle (\ED - \Sh)(u_A), (\Bz - \Sh)(u_A) \rangle = a \cdot \left(\frac{1}{n} - \frac{1}{a}\right)\left(\frac{1}{2^{a-1}} - \frac{1}{a}\right).
\]
Expanding, 
we get
\[
a\left(\frac{1}{n} - \frac{1}{a}\right)\left(\frac{1}{2^{a-1}} - \frac{1}{a}\right) = \frac{a}{n \cdot 2^{a-1}} - \frac{1}{n} - \frac{1}{2^{a-1}} + \frac{1}{a}.
\]
Summing $\binom{n}{a}$ times this over $a = 1, \ldots, n$, 
we obtain four sums:
\[
S_1 = \frac{1}{n}\sum_{a=1}^n \binom{n}{a}\frac{a}{2^{a-1}}, \quad
S_2 = \frac{2^n - 1}{n}, \quad
S_3 = \sum_{a=1}^n \binom{n}{a}\frac{1}{2^{a-1}}, \quad
S_4 = H_n.
\]
For $S_1$, 
using $a \binom{n}{a} = n \binom{n-1}{a-1}$,
\[
\sum_{a=1}^n \binom{n}{a}\frac{a}{2^a} = \frac{n}{2} \sum_{k=0}^{n-1} \binom{n-1}{k}\frac{1}{2^k} = \frac{n}{2}\left(\frac{3}{2}\right)^{n-1},
\]
so $S_1 = (2/n) \cdot (n/2)(3/2)^{n-1} = (3/2)^{n-1}$.
For $S_3$,
\[
S_3 = 2 \sum_{a=1}^n \binom{n}{a}\frac{1}{2^a} = 2\left( \left(\frac{3}{2}\right)^n - 1 \right).
\]
Combining,
\[
\langle \ED - \Sh, \Bz - \Sh \rangle_{\mL} = S_1 - S_2 - S_3 + S_4 = \left(\frac{3}{2}\right)^{n-1} - \frac{2^n - 1}{n} - 2\left(\frac{3}{2}\right)^n + 2 + H_n.
\]
Since $\left(\frac{3}{2}\right)^{n-1} - 2 \cdot \left(\frac{3}{2}\right)^n = -2 \cdot \left(\frac{3}{2}\right)^{n-1}$, and $H_n - (2^n-1)/n = D_n$,
 we obtain 
\[
\langle \ED - \Sh, \Bz - \Sh \rangle_{\mL} = D_n + 2 - 2\left(\frac{3}{2}\right)^{n-1}.
\]
Dividing by $D_n$ gives the stated formula.
\end{proof}

\begin{corollary}\label{cor-banzhaf-values}
For small $n$, the optimal $\eps^*(\Bz)$ takes the following values:
\[
\eps^*(\Bz)\Big|_{n=2} = 0, \qquad \eps^*(\Bz)\Big|_{n=3} = 0, \qquad \eps^*(\Bz)\Big|_{n=4} = \tfrac{1}{58}, \qquad \eps^*(\Bz)\Big|_{n=5} = \tfrac{11}{206}.
\]
\end{corollary}

\begin{proof}
For $n = 2$: $D_2 = 1$ and $2 - 2(3/2) = -1$, giving $\eps^* = 1 - 1 = 0$.
For $n = 3$: $D_3 = 5/2$ and $2 - 2(3/2)^2 = -5/2$, giving $\eps^* = 1 - 1 = 0$.
For $n = 4$: $D_4 = 29/6$ and $2 - 2(3/2)^3 
= -19/4$, 
giving $\eps^* = 1 - (19/4)/(29/6) = 1 - 57/58 = 1/58$.
For $n = 5$: $D_5 = 103/12$ and furthermore $2 - 2(3/2)^4 = 2 - 81/8 = -65/8$, giving $\eps^* = 1 - (65/8)/(103/12) = 1 - 195/206 = 11/206$. \qedhere
\end{proof}

\begin{remark}
For $n = 2$ and $n = 3$, the projection of the Banzhaf value onto the egalitarian Shapley line in fact lands exactly on the Shapley value. 
For $n = 2$, this is unsurprising since $\Bz \equiv \Sh$ on two-player games. 
For $n = 3$, the coincidence is more delicate: the Banzhaf value differs from the Shapley value, but the difference is $\langle \cdot, \cdot \rangle_{\mL}$-orthogonal to the line $\{\Feps\}$. 
From $n = 4$ onward, $\eps^*(\Bz)$ becomes positive but remains small, reflecting that the Banzhaf value is structurally close to, but slightly more egalitarian than, the Shapley value.
\end{remark}

\subsubsection*{The Equal-Surplus-Division Value}

The \emph{equal-surplus-division (ESD) value}~\cite{driessen1991} is given by
\[
\ESDi(v) := v(\{i\}) + \frac{1}{n}\left( v(N) - \sum_{j \in N} v(\{j\}) \right).
\]

\begin{proposition}\label{prop-eps-ESD}
\[
\eps^*(\ESD) = 1 - \frac{n-1}{D_n}.
\]
\end{proposition}

\begin{proof}
For the singleton unanimity game $u_{\{i_0\}}$ we have $u_{\{i_0\}}(\{i\}) = \delta_{i,i_0}$ and $u_{\{i_0\}}(N) = 1$, so
\[
\ESDi(u_{\{i_0\}}) = \delta_{i,i_0} + \frac{1}{n}\left(1 - 1\right) = \delta_{i,i_0} = \Shi(u_{\{i_0\}}).
\]
Hence, 
$(\ESD - \Sh)(u_{\{i_0\}}) = 0$.
On $u_A$ for $|A| \geq 2$, we have $u_A(\{i\}) = 0$ for all $i$, and $u_A(N) = 1$, and 
therefore it holds that $\ESDi(u_A) = 1/n = \EDi(u_A)$. 
Thus, $(\ESD - \Sh)(u_A) = (\ED - \Sh)(u_A)$ for $|A| \geq 2$.
It follows that
\[
\langle \ED - \Sh, \ESD - \Sh \rangle_{\mL} = \sum_{|A| \geq 2} \|(\ED - \Sh)(u_A)\|^2 = D_n - \sum_{a=1}^1 \binom{n}{a}\left(\frac{1}{a} - \frac{1}{n}\right) = D_n - (n - 1).
\]
Dividing by $D_n$ gives the result.
\end{proof}

\begin{corollary}\label{cor-ESD-values}
For small $n$, we have the following:
\[
\eps^*(\ESD)\Big|_{n=2} = 0, \quad \eps^*(\ESD)\Big|_{n=3} = \tfrac{1}{5}, \quad \eps^*(\ESD)\Big|_{n=4} = \tfrac{11}{29}, \quad \eps^*(\ESD)\Big|_{n=5} = \tfrac{55}{103}.
\]
Moreover, $\eps^*(\ESD) \to 1$ as $n \to \infty$.
\end{corollary}

\begin{proof}
For each $n$, substitute $D_n$ from Lemma~\ref{lem-Dn-values}. 
The limiting behavior follows from the fact that $D_n$ grows roughly as $H_n$, while the numerator $n - 1$ grows linearly: 
the singleton terms become a vanishing fraction of $D_n$.
\end{proof}

\subsection{Summary Table}

\begin{center}
\renewcommand{\arraystretch}{1.4}
\begin{tabular}{lcccc}
\hline
Target $\Psi$ & $n=2$ & $n=3$ & $n=4$ & $n=5$ \\
\hline
$\Sh$ (Shapley)               & $0$ & $0$ & $0$ & $0$ \\
$\ED$ (equal division)        & $1$ & $1$ & $1$ & $1$ \\
$\Bz$ (Banzhaf)               & $0$ & $0$ & $1/58$ & $11/206$ \\
$\ESD$ (equal-surplus div.)   & $0$ & $1/5$ & $11/29$ & $55/103$ \\
\hline
\end{tabular}
\end{center}

\section{Pythagorean Decomposition and Goodness of Fit}\label{sec-pythagorean}

The optimal projection $F^{\eps^*(\Psi)}$ produced by Theorem~\ref{thm-optimal-eps} suggests at this point the following natural decomposition of any linear value map:
into a component along the egalitarian  Shapley line and an orthogonal residual. 
This residual measures the part of $\Psi$ that lies \emph{outside} the marginalism--egalitarianism axis, 
and provides a quantitative goodness-of-fit  measure analogous to the coefficient of determination in linear regression.

\begin{figure}[h]
\centering
\begin{tikzpicture}[scale=1.15]

\draw[thick,gray] (-1.5,-0.6) -- (8.5,2.4);

\node[gray] at (7.6,2.6) {$\{F^\varepsilon \mid \varepsilon \in \mathbb{R}\}$};

\filldraw[black] (0.5,0.0) circle (1.6pt);
\node[left] at (0.60,0.15) {$\mathrm{Sh} = F^0$};

\filldraw[black] (5.5,1.5) circle (1.6pt);
\node[below right] at (5.40,1.55) {$\mathrm{ED} = F^1$};

\filldraw[black] (3.5,0.9) circle (1.6pt);
\node[below right, font=\small] at (3.55,0.95) {$F^{\varepsilon^*(\Psi)}$};

\filldraw[black] (2.9,2.9) circle (1.8pt);
\node[above] at (2.9,2.95) {$\Psi$};

\draw[dashed,gray] (0.5,0.0) -- (2.9,2.9);
\node[gray,left, font=\small] at (1.4,1.5) {$\Psi - \mathrm{Sh}$};

\draw[thick,->,>=latex] (0.5,0.0) -- (3.5,0.9);

\draw[thick,->,>=latex] (3.5,0.9) -- (2.9,2.9);

\draw[thick] 
  (3.26,0.828) -- (3.188,1.068) -- (3.428,1.14);

\node[below, font=\small] at (1.7,-0.1) {$\varepsilon^*(\Psi) \cdot (\mathrm{ED} - \mathrm{Sh})$};

\node[right,font=\small] at (3.25,1.9) {$\Phi_\perp(\Psi)$};

\end{tikzpicture}
\caption{The Pythagorean decomposition $\Psi - \mathrm{Sh} = \varepsilon^*(\Psi) \cdot (\mathrm{ED} - \mathrm{Sh}) + \Phi_\perp(\Psi)$. 
The optimal $F^{\varepsilon^*(\Psi)}$ is the orthogonal projection of $\Psi$ onto the affine line of egalitarian Shapley values. 
The residual $\Phi_\perp(\Psi)$ is orthogonal to the line.}
\label{fig-pythagorean}
\end{figure}
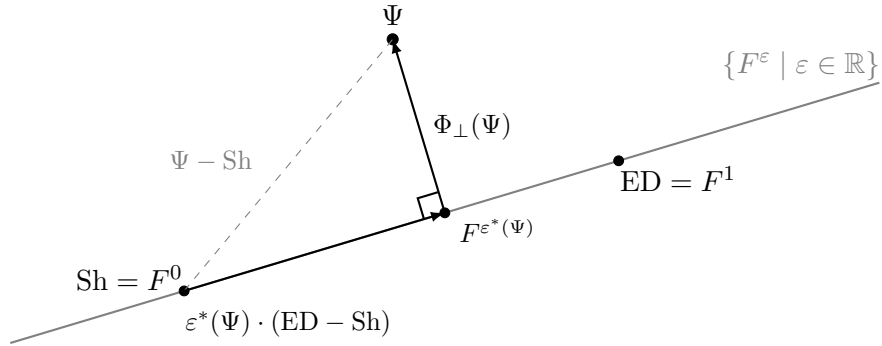

\begin{definition}\label{def-residual}
For $\Psi \in \mL$, the \emph{egalitarian Shapley residual of $\Psi$} is
\[
\Phi_\perp(\Psi) := \Psi - F^{\eps^*(\Psi)}.
\]
\end{definition}

\begin{lemma}\label{lem-orthogonality}
For all $\Psi \in \mL$,
\[
\langle \Phi_\perp(\Psi), \ED - \Sh \rangle_{\mL} = 0.
\]
\end{lemma}

\begin{proof}
Write $\Psi - \Sh = \eps^*(\Psi)(\ED - \Sh) + \Phi_\perp(\Psi)$. 
Then,
\begin{align*}
\langle \Phi_\perp(\Psi), \ED - \Sh \rangle_{\mL}
&= \langle \Psi - \Sh - \eps^*(\Psi)(\ED - \Sh), \ED - \Sh \rangle_{\mL} \\
&= \langle \Psi - \Sh, \ED - \Sh \rangle_{\mL} - \eps^*(\Psi) \cdot D_n.
\end{align*}
By Theorem~\ref{thm-optimal-eps}, $\eps^*(\Psi) = \langle \Psi - \Sh, \ED - \Sh \rangle_{\mL} / D_n$, 
so the right-hand side vanishes.
\end{proof}

\begin{theorem}[Pythagorean Decomposition]\label{thm-pythagorean}
For every $\Psi \in \mL$,
\[
\|\Psi - \Sh\|^2_{\mL} = (\eps^*(\Psi))^2 \cdot D_n + \|\Phi_\perp(\Psi)\|^2_{\mL}.
\]
\end{theorem}

\begin{proof}
By Lemma~\ref{lem-orthogonality},
 the two summands in $\Psi - \Sh = \eps^*(\Psi)(\ED - \Sh) + \Phi_\perp(\Psi)$ are orthogonal in $\langle \cdot, \cdot \rangle_{\mL}$. 
The Pythagoras theorem then implies that
\[
\|\Psi - \Sh\|^2_{\mL} = \|\eps^*(\Psi)(\ED - \Sh)\|^2_{\mL} + \|\Phi_\perp(\Psi)\|^2_{\mL} = (\eps^*(\Psi))^2 \cdot D_n + \|\Phi_\perp(\Psi)\|^2_{\mL}. \qedhere
\]
\end{proof}

\begin{definition}\label{def-Rsquared}
For $\Psi \in \mL$ with $\Psi \neq \Sh$,
 the \emph{egalitarian Shapley fit} of $\Psi$ is
\[
R^2(\Psi) := \frac{(\eps^*(\Psi))^2 \cdot D_n}{\|\Psi - \Sh\|^2_{\mL}}.
\]
We set $R^2(\Sh) := 1$ by convention.
\end{definition}

By Theorem~\ref{thm-pythagorean}, $R^2(\Psi) \in [0, 1]$ for every $\Psi \in \mL$, 
with $R^2(\Psi) = 1$ if and only if $\Psi$ lies on the egalitarian Shapley line $\{\Feps \mid  \eps \in \R\}$.

The interpretation of $R^2(\Psi) $ parallels the coefficient of determination from linear regression.
 Imagine the squared distance $\|\Psi - \Sh\|^2_{\mL}$ as a fixed budget describing how far $\Psi$ deviates from the Shapley value.
  The Pythagorean identity (Theorem~\ref{thm-pythagorean}) splits this budget into two orthogonal parts: 
  $(\eps^*(\Psi))^2 \cdot D_n$, the squared length of the projection along the line, and $\|\Phi_\perp(\Psi)\|^2_{\mL}$, the squared length of the residual perpendicular to it. 
  The quantity $R^2(\Psi)$ is the ratio of the first part to the total; 
  in other words, the fraction of $\Psi$'s deviation from $\Sh$ that can be \emph{explained} by movement along the egalitarian Shapley axis. 
  The complement $1 - R^2(\Psi)$ is the fraction left unexplained, lying in the residual.
When $R^2(\Psi)$ is close to $1$, almost all of $\Psi - \Sh$ points along the line $\{\Feps\}$, 
and $\Psi$ can be very well approximated by some egalitarian Shapley value. 
When $R^2(\Psi)$ is close to $0$, $\Psi - \Sh$ is essentially perpendicular to the line: 
$\Psi$ differs from $\Sh$ in directions that the egalitarian Shapley family cannot capture, 
and projecting $\Psi$ onto the line lands almost exactly back at $\Sh$ itself. 
Intermediate values of $R^2$ measure intermediate degrees of alignment.

\subsection{Norms of $\Psi - \Sh$ for Natural Targets}

We compute $\|\Psi - \Sh\|^2_{\mL}$ for the Banzhaf and equal-surplus-division values.

\begin{lemma}\label{lem-norm-Bz}
For $n \geq 2$,
\[
\|\Bz - \Sh\|^2_{\mL} = \sum_{a=1}^n \binom{n}{a} \cdot a \left(\frac{1}{2^{a-1}} - \frac{1}{a}\right)^2.
\]
\end{lemma}

\begin{proof}
For $|A| = a$, the vector $(\Bz - \Sh)(u_A)$ has components $1/2^{a-1} - 1/a$ for $i \in A$ and $0$ for $i \notin A$. 
  Hence,
\[
\|(\Bz - \Sh)(u_A)\|^2 = a \left(\frac{1}{2^{a-1}} - \frac{1}{a}\right)^2.
\]
Summing over all nonempty $A \subseteq N$ with multiplicity $\binom{n}{a}$ gives the formula.
\end{proof}

\begin{lemma}\label{lem-norm-ESD}
For $n \geq 2$, $\|\ESD - \Sh\|^2_{\mL} = D_n - (n-1)$.
\end{lemma}

\begin{proof}
By the proof of Proposition~\ref{prop-eps-ESD},  $(\ESD - \Sh)(u_A) = 0$ for $|A| = 1$ 
whereas we have $(\ESD - \Sh)(u_A) = (\ED - \Sh)(u_A)$ for $|A| \geq 2$. 
Thus,
 $\|\ESD - \Sh\|^2_{\mL} = \sum_{|A| \geq 2} \|(\ED - \Sh)(u_A)\|^2$, so
 that 
 $\|\ESD - \Sh\|^2_{\mL} = D_n - \sum_{a=1}^1 \binom{n}{a}(1/a - 1/n) = D_n - (n-1)$.
\end{proof}

\subsection{Numerical Results for $n = 4$}

We illustrate the framework by computing all relevant quantities for $n = 4$.

\begin{proposition}\label{prop-numerics-n4}
For $n = 4$, we have:
\begin{enumerate}[label=(\roman*)]
\item $\|\Bz - \Sh\|^2_{\mL} =  \frac{7}{48}$, 
$\|\Phi_\perp(\Bz)\|^2_{\mL} = \frac{67}{464}$, and  $R^2(\Bz) = \frac{2}{203} \approx 0.99\%$.
\item $\|\ESD - \Sh\|^2_{\mL} = \frac{11}{6} $, 
 $\|\Phi_\perp(\ESD)\|^2_{\mL} = \frac{33}{29}$, and $R^2(\ESD) = \frac{11}{29} \approx 37.9\%$.
\end{enumerate}
\end{proposition}

\begin{proof}
\emph{Part (i).} 
By Lemma~\ref{lem-norm-Bz}, $\|\Bz - \Sh\|^2_{\mL}$ is the sum, 
for $a = 1, \ldots, 4$, of $\binom{4}{a} \cdot a \cdot (1/2^{a-1} - 1/a)^2$. 
The terms are:
\begin{itemize}[leftmargin=2em]
\item $a = 1$: $1 \cdot 1 \cdot (1 - 1)^2 = 0$.
\item $a = 2$: $6 \cdot 2 \cdot (1/2 - 1/2)^2 = 0$.
\item $a = 3$: $4 \cdot 3 \cdot (1/4 - 1/3)^2 = 12 \cdot (1/12)^2 = 12/144 = 1/12$.
\item $a = 4$: $1 \cdot 4 \cdot (1/8 - 1/4)^2 = 4 \cdot (1/8)^2 = 4/64 = 1/16$.
\end{itemize}
Total: $1/12 + 1/16 = 4/48 + 3/48 = 7/48$.
By Corollary~\ref{cor-banzhaf-values}, $\eps^*(\Bz)|_{n=4} = 1/58$. 
By Lemma~\ref{lem-Dn-values}, $D_4 = 29/6$. 
The projection norm is therefore 
$(\eps^*(\Bz))^2 \cdot D_4 
= (1/58)^2 \cdot (29/6) 
= 1/(116 \cdot 6) = 1/696$.
By the Pythagorean identity (Theorem~\ref{thm-pythagorean}), 
$\|\Phi_\perp(\Bz)\|^2_{\mL} = \|\Bz - \Sh\|^2_{\mL} - (\eps^*(\Bz))^2 \cdot D_4$, 
which equals 
$7/48 - 1/696 = 67/464$. 
Finally, $R^2(\Bz) 
= (1/696) / (7/48) 
= 2/203$.

\medskip
\emph{Part (ii).} 
By Lemma~\ref{lem-norm-ESD}, 
we have $\|\ESD - \Sh\|^2_{\mL} = D_4 - 3 = 29/6 - 18/6 = 11/6$.
By Corollary~\ref{cor-ESD-values}, $\eps^*(\ESD)|_{n=4} = 11/29$. 
The projection norm is therefore 
equal to 
$(\eps^*(\ESD))^2 \cdot D_4 
= (121/841) \cdot (29/6) 
= 121/174$.
By the Pythagorean identity (Theorem~\ref{thm-pythagorean}) 
we obtain that  
$\|\Phi_\perp(\ESD)\|^2_{\mL} = \|\ESD - \Sh\|^2_{\mL} - (\eps^*(\ESD))^2 \cdot D_4 
= 11/6 - 121/174 = 33/29$. 
Finally, 
$R^2(\ESD) 
= (121/174)/(11/6) 
 = 11/29$.
\end{proof}

\begin{remark}\label{rem-Bz-orthogonal}
The result $R^2(\Bz) \approx 0.99\%$ for $n = 4$ has a striking interpretation: 
the Banzhaf value is \emph{nearly orthogonal} to the egalitarian Shapley axis. 
Although $\Bz \neq \Sh$, the difference $\Bz - \Sh$ lies almost entirely in directions orthogonal to $\ED - \Sh$ in $\mL$. 
The Banzhaf value therefore captures a kind of fairness consideration that is essentially independent of the marginalism--egalitarianism trade-off; 
it lives, geometrically, on a different axis altogether. 
This contrasts sharply with the equal-surplus-division value, where $R^2(\ESD) \approx 37.9\%$
 indicates substantial alignment with the egalitarian Shapley line.
\end{remark}

\begin{remark}\label{rem-explicit-residual}
The residual $\Phi_\perp(\Bz)$ admits an explicit description on unanimity games. 
Direct computation from $\Phi_\perp(\Bz) = \Bz - F^{\eps^*(\Bz)}$ yields, 
for each $\emptyset \neq A \subseteq N$ with $a = |A|$ and $i \in N$,
\[
\Phi_\perp(\Bz)_i(u_A) = 
\begin{cases}
\left(\dfrac{1}{2^{a-1}} - \dfrac{1}{a}\right) - \eps^*(\Bz) \left(\dfrac{1}{n} - \dfrac{1}{a}\right) & \text{if } i \in A, \\[0.4em]
- \eps^*(\Bz) \cdot \dfrac{1}{n} & \text{if } i \notin A.
\end{cases}
\]
While this gives a closed-form description of $\Phi_\perp(\Bz)$, 
the residual does not coincide with any standard named value map in the cooperative-game-theory literature: 
it is a generic element of the orthogonal complement of $\ED - \Sh$ in $\mL$, 
simply equal by construction to $\Bz - F^{\eps^*(\Bz)}$. 
The structural information about $\Phi_\perp(\Bz)$ is therefore captured by its norm $\|\Phi_\perp(\Bz)\|_{\mL}^2$ 
and the goodness-of-fit $R^2(\Bz)$, rather than by recognizing it as some other known object.
\end{remark}

\subsection{The Solidarity Value}\label{subsec-solidarity}

To complete the empirical picture, we compute the projection and goodness-of-fit 
also for the \emph{solidarity value} of Nowak and Radzik~\cite{nowak1994solidarity}, defined by
\[
\So_i(v) := \sum_{S \subseteq N,\, i \in S} \frac{(n - |S|)! \, (|S| - 1)!}{n!} \cdot A^v(S),
\]
where $A^v(S) := \frac{1}{|S|} \sum_{j \in S} (v(S) - v(S \setminus \{j\}))$ is the \emph{average marginal contribution} over members of $S$. 
The intuition is that when player $i$ joins a coalition, $i$ is credited with the average marginal 
contribution over all coalition members rather than $i$'s own marginal contribution alone. 
This is a structural form of solidarity. 
The proximity of $\So$ to other natural values has been studied informally by Radzik~\cite{radzik2013}.

\begin{lemma}\label{lem-So-unanimity}
For $\emptyset \neq A \subseteq N$, with $a = |A|$,
\[
A^{u_A}(S) =
 \begin{cases}
  \dfrac{a}{|S|} & \text{if } A \subseteq S, \\ 0 & \text{otherwise.} 
  \end{cases}
\]
\end{lemma}

\begin{proof}
If $A \not\subseteq S$, then $u_A(S) = 0$, 
and for each $j \in S$ also $A \not\subseteq S \setminus \{j\}$, so $u_A(S \setminus \{j\}) = 0$. 
Thus, $A^{u_A}(S) = 0$.
If $A \subseteq S$, then $u_A(S) = 1$. 
For $j \in S \setminus A$: $A \subseteq S \setminus \{j\}$, so $u_A(S \setminus \{j\}) = 1$, 
and the marginal contribution is $0$. 
For $j \in A$: $A \not\subseteq S \setminus \{j\}$, so $u_A(S \setminus \{j\}) = 0$, 
and the marginal contribution is $1$. 
Summing over $j \in S$ gives $|A| = a$, 
and dividing by $|S|$ gives $A^{u_A}(S) = a/|S|$.
\end{proof}

We now compute $\So(u_A)$ explicitly for $n = 4$ and each subset size $a \in \{1, 2, 3, 4\}$.

\begin{lemma}\label{lem-So-n4}
For $n = 4$ and any nonempty $A \subseteq N$ with $a = |A|$,
\[
(\So - \Sh)(u_A)_i = \begin{cases}
0 & \text{if } a = 4, \\
-1/16 \text{ for } i \in A, \quad 3/16 \text{ for } i \notin A & \text{if } a = 3, \\
-13/72 \text{ for } i \in A, \quad 13/72 \text{ for } i \notin A & \text{if } a = 2, \\
-23/48 \text{ for } i \in A, \quad 23/144 \text{ for } i \notin A & \text{if } a = 1.
\end{cases}
\]
\end{lemma}

\begin{proof}
We compute $\So_i(u_A)$ from the definition by splitting into the cases $i \in A$ and $i \notin A$.

\smallskip
\emph{Case $a = 4$ ($A = N$).} 
Only the term $S = N$ contributes (Lemma~\ref{lem-So-unanimity}), 
giving   the equation $\So_i(u_N) = (0! \cdot 3!/4!) \cdot 1 = 1/4 = \Sh_i(u_N)$ for each $i$.

\smallskip
\emph{Case $a = 3$.} 
For $i \in A$: contributing $S$ have $|S| \in \{3, 4\}$, with multiplicity $\binom{1}{s-3}$. 
The contributions are $(1!\cdot 2!/4!)\cdot 1 = 1/12$ at $s=3$ and $(0!\cdot 3!/4!)\cdot (3/4) = 3/16$ at $s=4$, 
summing to $1/12 + 3/16 = 13/48$. 
Subtracting $\Sh_i(u_A) = 1/3$ gives $-1/16$. 
For $i \notin A$: only $S = N$ contributes ($1/4 \cdot 3/4 = 3/16$), 
and $\Sh_i(u_A) = 0$.

\smallskip
\emph{Case $a = 2$.} 
For $i \in A$: contributing $S$ have $|S| \in \{2, 3, 4\}$ with multiplicities $\binom{2}{s-2}$. 
Their contributions: $(2!\cdot 1!/4!)\cdot 1 = 1/12$ at $s=2$, 
$2 \cdot (1!\cdot 2!/4!) \cdot (2/3) = 1/9$ at $s=3$, 
and finally $(0!\cdot 3!/4!) \cdot (1/2) = 1/8$ at $s=4$.  
Total: $1/12 + 1/9 + 1/8 = 23/72$. 
Subtracting $\Sh_i(u_A) = 1/2$ gives $-13/72$. 
For $i \notin A$: contributions $1/18$ at $s=3$ and $1/8$ at $s=4$, summing to $13/72$.

\smallskip
\emph{Case $a = 1$.} 
For $i \in A$: contributions $1/4$ at $s=1$, $1/8$ at $s=2$, $1/12$ at $s=3$, $1/16$ at $s=4$, totalling $25/48$. 
Subtracting $\Sh_i(u_A) = 1$ gives $-23/48$. 
For $i \notin A$: contributions $1/24$ at $s=2$, $1/18$ at $s=3$, $1/16$ at $s=4$, totalling $23/144$.

In every case the components sum to zero, confirming consistency with efficiency.
\end{proof}

\begin{proposition}\label{prop-So-projection}
For $n = 4$, the solidarity value satisfies
\[
\eps^*(\So) = \frac{39}{58}, \qquad \|\So - \Sh\|^2_{\mL} = \frac{79}{36}, \qquad \|\Phi_\perp(\So)\|^2_{\mL} = \frac{19}{2088}, \qquad R^2(\So) = \frac{4563}{4582} \approx 99.6\%.
\]
\end{proposition}

\begin{proof}
We compute $\langle \ED - \Sh, \So - \Sh \rangle_{\mL}$ as a sum over subset sizes, 
using the inner product on $\R^4$ between the components in Lemma~\ref{lem-So-n4} 
and the components $(\ED - \Sh)_i(u_A) = 1/4 - 1/a$ for $i \in A$, $1/4$ for $i \notin A$.

\medskip
\emph{Size $a = 4$:} contribution $0$ (since $(\So - \Sh)(u_N) = 0$).

\emph{Size $a = 3$:} per subset, $3 \cdot (-1/12)(-1/16) + 1 \cdot (1/4)(3/16) = 3/192 + 3/64 = 12/192 = 1/16$. 
Times $\binom{4}{3} = 4$: total $1/4$.

\emph{Size $a = 2$:} per subset, $2 \cdot (-1/4)(-13/72) + 2 \cdot (1/4)(13/72) = 26/288 + 26/288 = 13/72$. 
Times $\binom{4}{2} = 6$: total $13/12$.

\emph{Size $a = 1$:} per subset, $1 \cdot (-3/4)(-23/48) + 3 \cdot (1/4)(23/144) =  207/576 + 69/576 = 23/48$. 
Times $\binom{4}{1} = 4$: total $23/12$.

\medskip
Summing, we get $\langle \ED - \Sh, \So - \Sh \rangle_{\mL} = 0 + 1/4 + 13/12 + 23/12 = 39/12 = 13/4$. 
Therefore,
\[
\eps^*(\So) = \frac{13/4}{D_4} = \frac{13/4}{29/6} = \frac{13 \cdot 6}{4 \cdot 29} = \frac{39}{58}.
\]

\medskip
For $\|\So - \Sh\|^2_{\mL}$, 
summing $\binom{n}{a} \cdot \|(\So - \Sh)(u_A)\|^2$ over $a$:
\begin{align*}
a = 3: \quad & \binom{4}{3} \cdot [3 \cdot (1/16)^2 + (3/16)^2] = 4 \cdot 12/256 = 3/16, \\
a = 2: \quad & \binom{4}{2} \cdot 4 \cdot (13/72)^2 = 6 \cdot 676/5184 = 169/216, \\
a = 1: \quad & \binom{4}{1} \cdot [(23/48)^2 + 3 \cdot (23/144)^2] = 4 \cdot 529/1728 = 529/432.
\end{align*}
Total: $3/16 + 169/216 + 529/432 = 81/432 + 338/432 + 529/432 = 948/432 = 79/36$.

\medskip
The projection norm is 
$(\eps^*(\So))^2 \cdot D_4 = (39/58)^2 \cdot (29/6) = 507/232$. 
 By the Pythagorean identity (Theorem~\ref{thm-pythagorean}),
\[
\|\Phi_\perp(\So)\|^2_{\mL} = \|\So - \Sh\|^2_{\mL} - (\eps^*(\So))^2 \cdot D_4 = \frac{79}{36} - \frac{507}{232}.
\]
We obtain that 
$\|\Phi_\perp(\So)\|^2_{\mL} = 19/2088$.
Finally, 
$R^2(\So) 
= (507/232) / (79/36) 
= 4563/4582$.
\end{proof}

\begin{remark}\label{rem-So-aligned}
The result $R^2(\So) \approx 99.6\%$ for $n = 4$ is the most striking finding of this section. 
The solidarity value lies almost entirely on the egalitarian Shapley axis: 
only about $0.4\%$ of its squared distance from the Shapley value is in the orthogonal residual. 
Geometrically, this means that for $n = 4$ the solidarity value is essentially indistinguishable from a member of the egalitarian Shapley family 
(specifically, very close to $F^{39/58}$)
 with $\eps^* = 39/58 \approx 0.67$ placing it well past halfway from $\Sh$ towards $\ED$. 
 This is interpretable in retrospect: 
 the Nowak--Radzik construction replaces individual marginal contributions by their averages over coalition members, 
 which is structurally an egalitarian move, 
 and our framework reveals that this structural egalitarianism manifests almost purely along the marginalism--egalitarianism axis.
\end{remark}

\subsection{A Geometric Classification}\label{subsec-classification}

Combining the four targets considered, 
we obtain a clean geometric classification of the standard alternatives to the Shapley value at $n = 4$:
\begin{itemize}[leftmargin=2em]
\item The \emph{Banzhaf value} is \emph{nearly orthogonal} to the egalitarian Shapley axis ($R^2 \approx 1\%$):   
   indeed its difference from $\Sh$ lies almost entirely in directions independent of the marginalism--egalitarianism trade-off.
\item The \emph{equal-surplus-division value} is \emph{moderately aligned} with the axis ($R^2 \approx 38\%$):
 indeed its difference from $\Sh$ has substantial components both along and orthogonal to the line $\{\Feps\}$.
\item The \emph{solidarity value} is \emph{almost entirely aligned} with the axis ($R^2 \approx 99.6\%$):
 it is, geometrically, an egalitarian Shapley value in disguise.
\end{itemize}

\begin{center}
\renewcommand{\arraystretch}{1.4}
\begin{tabular}{lccccc}
\hline
$\Psi$ ($n=4$) & $\eps^*(\Psi)$ & $\|\Psi-\Sh\|^2_{\mL}$ & $(\eps^*(\Psi))^2 D_n$ & $\|\Phi_\perp(\Psi)\|^2_{\mL}$ & $R^2(\Psi)$ \\
\hline
$\Sh$    & $0$     & $0$    & $0$       & $0$        & $1$ (convention) \\
$\ED$    & $1$     & $29/6$ & $29/6$    & $0$        & $1$ \\
$\Bz$    & $1/58$  & $7/48$ & $1/696$   & $67/464$   & $2/203 \approx 0.99\%$ \\
$\ESD$   & $11/29$ & $11/6$ & $121/174$ & $33/29$    & $11/29 \approx 37.9\%$ \\
$\So$    & $39/58$ & $79/36$ & $507/232$ & $19/2088$  & $4563/4582 \approx 99.6\%$ \\
\hline
\end{tabular}
\end{center}

\medskip

The classification  suggests that the space of  fairness concepts embodied by 
linear value maps decomposes into at least two genuinely independent dimensions: 
a marginalism--egalitarianism dimension, captured by the line $\{\Feps\}$, 
and at least one further dimension orthogonal to it. 
Banzhaf-style power-index reasoning lives primarily on this orthogonal dimension. 
Solidarity-style averaging, despite its different motivation,  lives almost entirely along the marginalism--egalitarianism axis. 
The equal-surplus-division  value sits as a proper blend of the two. 
This is, to our knowledge, the first quantitative articulation of the well-known but informal sense
 (also discussed in~\cite{kamijokongo2012,casajus2014}) that the Banzhaf value ``feels different'' from 
 the Shapley family in a way that the equal-surplus-division and solidarity values do not.

\section{Spectral Decomposition of the Egalitarian Shapley Flat}\label{sec-spectral}

The Pythagorean decomposition of Section~\ref{sec-pythagorean} treats each linear value map $\Psi$ as a single point in $\mL$ and projects it onto the one-dimensional egalitarian Shapley line. 
We now refine this picture by exploiting an orthogonal stratification of $\mL$ by coalition size. 
 The main result of this section is a structural one: every efficient, symmetric, linear value map admits a canonical 
 decomposition into $n - 1$ \emph{stratified epsilons},   one per coalition size, 
 and the entire space of such value maps is parametrized linearly by $\R^{n-1}$. 
 The classical egalitarian Shapley line $\{\Feps \mid \eps \in \R\}$ becomes the diagonal $\{(\eps, \ldots, \eps)\}$ in this $\R^{n-1}$. 
  The goodness-of-fit measure $R^2$ then admits   a closed-form interpretation as one minus the relative weighted variance of the stratified epsilons.

\subsection{Size Strata of \texorpdfstring{$\mL$}{L}}

For each $a \in \{1, \ldots, n\}$, write
\[
\GN_a := \mathrm{span}\{u_A \mid |A| = a\} \subset \GN,
\]
so that $\GN = \bigoplus_{a=1}^n \GN_a$ orthogonally with respect to the Harsanyi inner product, and 
furthermore  $\dim \GN_a = \binom{n}{a}$.

\begin{definition}\label{def-stratum}
The \emph{size-$a$ stratum} of $\mL$ is
\[
\mL_a := \{\Psi \in \mL \mid \Psi(u_B) = 0 \text{ for all } B \subseteq N \text{ with } |B| \neq a\}.
\]
Equivalently, $\mL_a$ is the subspace of linear maps that act trivially outside $\GN_a$.  
The \emph{stratum projection} is the linear map $\Psi \mapsto \Psi^{(a)} \in \mL_a$ defined on the unanimity basis 
(for every nonempty $A \subseteq N$) by
\[
\Psi^{(a)}(u_A) =
 \begin{cases}
 \Psi(u_A) & \text{if } |A| = a, \\ 0 & \text{if } |A| \neq a. 
\end{cases}
\]
\end{definition}

\begin{lemma}\label{lem-stratum-orthogonal}
$\mL = \bigoplus_{a=1}^n \mL_a$ orthogonally with respect to $\langle \cdot, \cdot \rangle_\mL$. 
 For all $\Psi, \Phi \in \mL$,
\[
\langle \Psi, \Phi \rangle_\mL = \sum_{a=1}^n \langle \Psi^{(a)}, \Phi^{(a)} \rangle_\mL.
\]
\end{lemma}

\begin{proof}
By Definition~\ref{def-IP-on-L}, $\langle \Psi, \Phi \rangle_\mL = \sum_A \langle \Psi(u_A), \Phi(u_A) \rangle_\RN$, 
summed over nonempty $A \subseteq N$. 
If $\Psi \in \mL_a$ and $\Phi \in \mL_b$ with $a \neq b$, every term in this sum is zero, 
since each $A$ has a definite size and either $\Psi(u_A) = 0$ or $\Phi(u_A) = 0$. 
  Hence the strata are mutually orthogonal.
  Every $\Psi \in \mL$ decomposes as $\Psi = \sum_a \Psi^{(a)}$ since $\{u_A\}$ 
  spans $\GN$ and $\Psi$ is determined by its values on this basis.
\end{proof}

\subsection{The Symmetric Strata and the Efficient Direction}\label{subsec-sym-strata}

We restrict now to the subspace $\mL^{\mathrm{S}} \subseteq \mL$ of symmetric linear value maps,
  i.e., those satisfying axiom (Sym) of Section~\ref{sec-coop-games}, 
   and consider the symmetric strata $\mL_a^{\mathrm{S}} := \mL_a \cap \mL^{\mathrm{S}}$.

\begin{lemma}\label{lem-sym-stratum-dim}
For $\Psi \in \mL^{\mathrm{S}}$ and any $A \subseteq N$ with $|A| = a$, 
the vector $\Psi(u_A) \in \RN$ has equal entries on $A$ and equal entries on $N \setminus A$. 
  Consequently, $\Psi^{(a)}$ is determined by two scalars
\[
\alpha_a(\Psi) := \Psi_i(u_A) \text{ for any } i \in A, \qquad \beta_a(\Psi) := \Psi_i(u_A) \text{ for any } i \in N \setminus A,
\]
when $a < n$, and by a single scalar $\alpha_n(\Psi) = \Psi_i(u_N)$ when $a = n$. 
Therefore, 
\[
\dim \mL_a^{\mathrm{S}} =
  \begin{cases}
   2 & \text{if } 1 \leq a \leq n - 1, \\ 1 & \text{if } a = n.
 \end{cases}
\]
\end{lemma}

\begin{proof}
For $A \subseteq N$ with $|A| = a$, any two players in $A$ are $u_A$-interchangeable, 
as are any two players in $N \setminus A$, 
since $u_A(S \cup \{i\})$ depends on $S$ and $A$ only through $|S \cap A|$ and the membership of $i$ in $A$. 
 The symmetry axiom (Sym) then forces $\Psi_i(u_A)$ to be constant on $A$ and constant on $N \setminus A$. 
  For different $A$ of the same size, transposing within $S_n$ gives $\Psi_{\pi(i)}(u_{\pi(A)}) = \Psi_i(u_A)$ 
  for any permutation $\pi$, so $\alpha_a$ and $\beta_a$ depend only on $a$, not on the choice of $A$.
    The dimension count follows immediately.
\end{proof}

The inner product on stratum $\mL_a^{\mathrm{S}}$ admits a clean expression in $(\alpha_a, \beta_a)$ coordinates:
\begin{equation}\label{eq-stratum-IP}
\langle \Psi^{(a)}, \Phi^{(a)} \rangle_\mL = \binom{n}{a} \cdot \big[a \, \alpha_a(\Psi) \alpha_a(\Phi) + (n-a) \, \beta_a(\Psi) \beta_a(\Phi)\big],
\end{equation}
for $a < n$, and $\langle \Psi^{(n)}, \Phi^{(n)} \rangle_\mL = n \, \alpha_n(\Psi) \alpha_n(\Phi)$.

We now identify two natural orthogonal directions inside each stratum.

\begin{definition}\label{def-eff-unif}
For $a < n$, the \emph{efficient subspace} and the \emph{uniform subspace} of $\mL_a^{\mathrm{S}}$ are
\begin{align*}
\mL_a^{\mathrm{eff}} &:= \{\Psi^{(a)} \in \mL_a^{\mathrm{S}} \mid a\, \alpha_a + (n-a)\, \beta_a = 0\}, \\
\mL_a^{\mathrm{unif}} &:= \{\Psi^{(a)} \in \mL_a^{\mathrm{S}} \mid \alpha_a = \beta_a\}.
\end{align*}
\end{definition}
The constraint $a\, \alpha_a + (n-a)\, \beta_a = 0$ characterizes the maps with vanishing total worth on $u_A$, 
 hence the term ``efficient'': differences of efficient value maps lie there.
   The constraint $\alpha_a = \beta_a$ characterizes maps that pay every player the same on $u_A$,
    regardless of membership in $A$,
    and is the direction along which the equal-division solution moves.

\begin{lemma}\label{lem-eff-unif-orthogonal}
For $a < n$, the subspaces $\mL_a^{\mathrm{eff}}$ and $\mL_a^{\mathrm{unif}}$
  are each one-dimensional and mutually orthogonal in $\langle \cdot, \cdot \rangle_\mL$, 
   and
\[
\mL_a^{\mathrm{S}} = \mL_a^{\mathrm{eff}} \oplus \mL_a^{\mathrm{unif}}.
\]
\end{lemma}

\begin{proof}
Each subspace is cut out by a single linear constraint in the two-dimensional space $\mL_a^{\mathrm{S}}$, 
 hence has dimension one.
   For orthogonality, take $\Psi^{(a)} \in \mL_a^{\mathrm{eff}}$ and $\Phi^{(a)} \in \mL_a^{\mathrm{unif}}$.
    By~\eqref{eq-stratum-IP},
\[
\langle \Psi^{(a)}, \Phi^{(a)} \rangle_\mL = \binom{n}{a} \big[a \, \alpha_a(\Psi) \alpha_a(\Phi) + (n-a)\, \beta_a(\Psi) \beta_a(\Phi)\big].
\]
Since $\Phi^{(a)} \in \mL_a^{\mathrm{unif}}$, it holds that $\alpha_a(\Phi) = \beta_a(\Phi)$ (call this common value $c$); since $\Psi^{(a)} \in \mL_a^{\mathrm{eff}}$, $a\, \alpha_a(\Psi) + (n-a)\, \beta_a(\Psi) = 0$.
 Substituting,
\[
\langle \Psi^{(a)}, \Phi^{(a)} \rangle_\mL = \binom{n}{a} \cdot c \cdot \big[a\, \alpha_a(\Psi) + (n-a)\, \beta_a(\Psi)\big] = 0.
\]
The two one-dimensional subspaces span the two-dimensional $\mL_a^{\mathrm{S}}$ since they are not equal: 
 the line $\alpha = \beta$ does not satisfy $a \alpha + (n-a) \beta = 0$ except at the origin,
  so the intersection is trivial.
\end{proof}

\begin{remark}\label{rem-Sh-ED-coords}
The Shapley value and equal-division solution have stratum coordinates
\[
\alpha_a(\Sh) = \tfrac{1}{a}, \quad \beta_a(\Sh) = 0, \qquad \alpha_a(\ED) = \beta_a(\ED) = \tfrac{1}{n}.
\]
Hence $\ED^{(a)}$ lies entirely in $\mL_a^{\mathrm{unif}}$, while $(\ED - \Sh)^{(a)}$ lies entirely in $\mL_a^{\mathrm{eff}}$ 
(as one verifies readily: $a(\frac{1}{n}-\frac{1}{a}) + (n-a)\frac{1}{n} = \frac{a}{n} - 1 + \frac{n-a}{n} = 0$). 
The egalitarian Shapley line $\{\Feps\}$ moves $\Sh$ along $\mL_a^{\mathrm{eff}}$ 
at every stratum simultaneously and by the same multiple $\eps$.
\end{remark}

\subsection{Stratified Epsilons for Efficient Symmetric Linear Value Maps}\label{subsec-stratified-eps}

Throughout this subsection, let 
$\mL^{\mathrm{ESL}} := \mL^{\mathrm{S}} \cap \{\text{efficient}\} \cap \{\text{linear}\}$ 
denote the space of efficient symmetric linear value maps. 
(Linearity is implicit in $\mL$;  
  we include it in the notation for clarity.)

\begin{lemma}\label{lem-Eff-implies-eff-stratum}
If $\Psi \in \mL^{\mathrm{ESL}}$, then $(\Psi - \Sh)^{(a)} \in \mL_a^{\mathrm{eff}}$ for every $a < n$.
\end{lemma}

\begin{proof}
By efficiency of both $\Psi$ and $\Sh$, 
$\sum_{i \in N} \Psi_i(u_A) = u_A(N) = \sum_{i \in N} \Sh_i(u_A)$
  for every $A$, so $\sum_i (\Psi - \Sh)_i(u_A) = 0$.
   For $|A| = a < n$, this reads $a (\alpha_a(\Psi) - \alpha_a(\Sh)) + (n-a)(\beta_a(\Psi) - \beta_a(\Sh)) = 0$,
    which is exactly the efficient-direction condition.
\end{proof}

Since $\mL_a^{\mathrm{eff}}$ is one-dimensional and contains the nonzero vector $(\ED - \Sh)^{(a)}$ 
(which is nonzero for $a < n$ because $1/a \neq 1/n$),
 there is a unique scalar $\eps_a(\Psi) \in \R$ satisfying
\begin{equation}\label{eq-eps-a-defn}
(\Psi - \Sh)^{(a)} = \eps_a(\Psi) \cdot (\ED - \Sh)^{(a)} \quad \text{in } \mL_a^{\mathrm{eff}}.
\end{equation}
For $a = n$, any efficient symmetric value map satisfies $\alpha_n = 1/n$, 
 so $\Psi^{(n)} = \Sh^{(n)} = \ED^{(n)}$ and the analogous coordinate is degenerate; 
 we set $\eps_n(\Psi) := 1$ by convention.

\begin{theorem}[Stratified Egalitarian Shapley Flat]\label{thm-stratified-flat}
The map
\[
\Phi_{\mathrm{strat}} \colon \mL^{\mathrm{ESL}} \longrightarrow \R^{n-1}, \qquad \Psi \longmapsto (\eps_1(\Psi), \ldots, \eps_{n-1}(\Psi)),
\]
is a linear isomorphism. 
Equivalently, every efficient symmetric linear value map $\Psi$ admits a unique stratum-by-stratum decomposition
\[
\Psi^{(a)} = (1 - \eps_a(\Psi)) \cdot \Sh^{(a)} + \eps_a(\Psi) \cdot \ED^{(a)} \qquad (1 \leq a \leq n - 1),
\]
with $\Psi^{(n)} = \Sh^{(n)} = \ED^{(n)}$ forced. 
The classical egalitarian Shapley value $\Feps$ corresponds to the diagonal $\eps_1(\Feps) = \cdots = \eps_{n-1}(\Feps) = \eps$.
\end{theorem}

\begin{proof}
Linearity of $\Phi_{\mathrm{strat}}$: each $\eps_a(\Psi)$ is a linear functional of $\Psi$ by Lemma~\ref{lem-Eff-implies-eff-stratum} 
and the one-dimensionality of $\mL_a^{\mathrm{eff}}$ (the coefficient on $(\ED - \Sh)^{(a)}$ is linear in the coordinates of $(\Psi - \Sh)^{(a)}$).

Injectivity: if $\eps_a(\Psi) = 0$ for all $a < n$, then $(\Psi - \Sh)^{(a)} = 0$ for $a = 1, \ldots, n-1$, 
and $(\Psi - \Sh)^{(n)} = 0$ since $\Psi$ and $\Sh$ are both efficient and symmetric. 
By Lemma~\ref{lem-stratum-orthogonal}, 
we have $\Psi - \Sh = \sum_a (\Psi - \Sh)^{(a)} = 0$, i.e., $\Psi = \Sh$.

Surjectivity: given $(\eps_1, \ldots, \eps_{n-1}) \in \R^{n-1}$, define $\Psi$ by setting 
$\Psi^{(a)} := (1 - \eps_a) \Sh^{(a)} + \eps_a \ED^{(a)}$ for $a = 1, \ldots, n-1$ 
and $\Psi^{(n)} := \Sh^{(n)}$, and extending by linearity. 
On each unanimity game $u_A$ with $|A| = a < n$, the construction gives $\Psi(u_A) = (1 - \eps_a) \Sh(u_A) + \eps_a \ED(u_A)$, 
whose entries sum to $(1 - \eps_a) \cdot 1 + \eps_a \cdot 1 = 1 = u_A(N)$. 
On $u_N$, $\Psi(u_N) = \Sh(u_N) = (1/n, \ldots, 1/n)$, also summing to $1$. 
Hence $\Psi$ is efficient on every unanimity game, and by linearity efficient on every game in $\GN$. 
Symmetry of each $\Psi^{(a)}$ (Lemma~\ref{lem-sym-stratum-dim}) extends to symmetry of $\Psi$. 
By construction, $\eps_a(\Psi) = \eps_a$ for each $a$.

The classical egalitarian Shapley value $\Feps = (1-\eps)\Sh + \eps \ED$ has $(\Feps)^{(a)} = (1-\eps)\Sh^{(a)} + \eps \ED^{(a)}$ for every $a$, so $\eps_a(\Feps) = \eps$ for all $a$.
\end{proof}

Theorem~\ref{thm-stratified-flat} reframes the geometry of the entire paper. 
The space of efficient symmetric linear value maps is not an opaque high-dimensional object:
 it is canonically $\R^{n-1}$, and every such value map is a \emph{stratified egalitarian Shapley value}, 
 with one parameter $\eps_a$ per coalition size.
  The classical one-parameter family is exactly the diagonal slice.

\subsection{The Pythagorean Decomposition Revisited}\label{subsec-pythag-revisited}

We now express the global quantities $\eps^*(\Psi)$ and $R^2(\Psi)$ from 
Sections~\ref{sec-best-approx} and~\ref{sec-pythagorean} in terms of the stratified epsilons.

\begin{definition}\label{def-stratum-weights}
For $1 \leq a \leq n - 1$, the \emph{stratum weight} is
\[
w_a := \frac{\|(\ED - \Sh)^{(a)}\|^2_\mL}{D_n} = \frac{\binom{n}{a} \cdot (1/a - 1/n)}{D_n}.
\]
\end{definition}

By the formula for $D_n = \sum_{a=1}^{n-1} \binom{n}{a}(1/a - 1/n)$ from Section~\ref{sec-best-approx} (the $a = n$ term vanishes),
 the weights satisfy $\sum_{a=1}^{n-1} w_a = 1$, so $w$ defines a probability distribution on $\{1, \ldots, n-1\}$.

\begin{theorem}[Stratified Projection Formula]\label{thm-stratified-proj}
For every $\Psi \in \mL^{\mathrm{ESL}}$,
\[
\eps^*(\Psi) = \sum_{a=1}^{n-1} w_a \cdot \eps_a(\Psi) = \E_w[\eps_a(\Psi)],
\]
where $\E_w$ denotes expectation under the weights $\{w_a\}$.
\end{theorem}

\begin{proof}
By Lemma~\ref{lem-stratum-orthogonal} and \eqref{eq-eps-a-defn},
\[
\langle \ED - \Sh, \Psi - \Sh \rangle_\mL = \sum_{a=1}^{n-1} \langle (\ED - \Sh)^{(a)}, (\Psi - \Sh)^{(a)} \rangle_\mL = \sum_{a=1}^{n-1} \eps_a(\Psi) \cdot \|(\ED - \Sh)^{(a)}\|^2_\mL.
\]
By Theorem~\ref{thm-optimal-eps},
\[
\eps^*(\Psi) = \frac{\langle \ED - \Sh, \Psi - \Sh \rangle_\mL}{D_n} = \sum_{a=1}^{n-1} \eps_a(\Psi) \cdot \frac{\|(\ED - \Sh)^{(a)}\|^2_\mL}{D_n} = \sum_{a=1}^{n-1} w_a \cdot \eps_a(\Psi). \qedhere
\]
\end{proof}

\begin{remark}\label{rem-stratified-proj-general}
The proof of Theorem~\ref{thm-stratified-proj} does not use efficiency of $\Psi$. 
The same identity $\eps^*(\Psi) = \E_w[\eps_a(\Psi)]$ holds for every $\Psi \in \mL^{\mathrm{S}}$, with $\eps_a(\Psi)$ 
defined as the coefficient of $(\Psi - \Sh)^{(a)}$ along $(\ED - \Sh)^{(a)}$ in the orthogonal decomposition 
of $\mL_a^{\mathrm{S}}$ (see~\eqref{eq-eps-delta-decomp} below).
 This more general version is used implicitly in Section~\ref{subsec-Bz-stratified} when treating the Banzhaf value.
\end{remark}

\begin{theorem}[$R^2$ as One Minus Relative Variance]\label{thm-R2-variance}
For every $\Psi \in \mL^{\mathrm{ESL}}$ with $\Psi \neq \Sh$,
\[
R^2(\Psi) = \frac{\E_w[\eps_a(\Psi)]^2}{\E_w[\eps_a(\Psi)^2]} = 1 - \frac{\Var_w(\eps_a(\Psi))}{\E_w[\eps_a(\Psi)^2]},
\]
where $\Var_w(\eps_a) := \E_w[\eps_a^2] - \E_w[\eps_a]^2$.
 In particular, $R^2(\Psi) = 1$ if and only if $\eps_a(\Psi)$ is constant in $a$, i.e., 
 $\Psi$ lies on the egalitarian Shapley line $\{\Feps\}$.
\end{theorem}

\begin{proof}
By Lemma~\ref{lem-stratum-orthogonal} and \eqref{eq-eps-a-defn},
\[
\|\Psi - \Sh\|^2_\mL = \sum_{a=1}^{n-1} \|(\Psi - \Sh)^{(a)}\|^2_\mL = \sum_{a=1}^{n-1} \eps_a(\Psi)^2 \cdot \|(\ED - \Sh)^{(a)}\|^2_\mL = D_n \cdot \E_w[\eps_a(\Psi)^2].
\]
By Theorem~\ref{thm-stratified-proj}, $(\eps^*(\Psi))^2 \cdot D_n = D_n \cdot \E_w[\eps_a(\Psi)]^2$. 
Therefore, we have
\[
R^2(\Psi) = \frac{(\eps^*(\Psi))^2 \cdot D_n}{\|\Psi - \Sh\|^2_\mL} = \frac{D_n \cdot \E_w[\eps_a]^2}{D_n \cdot \E_w[\eps_a^2]} = \frac{\E_w[\eps_a]^2}{\E_w[\eps_a^2]} = 1 - \frac{\Var_w(\eps_a)}{\E_w[\eps_a^2]}.
\]
The last claim follows because $\Var_w(\eps_a) = 0$
 if and only if 
 $\eps_a$ is $w$-almost-surely constant in $a$, which is 
 equivalent (since each $w_a > 0$) to $\eps_a$ being constant.
\end{proof}

Theorem~\ref{thm-R2-variance} is the structural analogue of the regression coefficient of determination: 
$R^2(\Psi)$ is the squared ratio of the weighted mean of the stratified epsilons to their weighted root-mean-square, 
and equals one exactly when the epsilons are constant across strata. 
It also gives a sharp diagnostic: $1 - R^2(\Psi)$ measures
  \emph{how far the stratified epsilons spread around their mean}, normalized by their second moment.

\subsection{Examples Recovered}\label{subsec-examples}

Theorem~\ref{thm-R2-variance} reproduces every numerical computation of Section~\ref{sec-pythagorean} as a single line of arithmetic. 
We illustrate.

\begin{example}[Equal-Surplus-Division Value]\label{ex-ESD-stratified}
On unanimity games, $\ESD$ agrees with $\Sh$ on singletons and with $\ED$ on larger coalitions: $\ESD(u_{\{i_0\}}) = e_{i_0} = \Sh(u_{\{i_0\}})$ gives $\eps_1(\ESD) = 0$, 
while $\ESD(u_A) = (1/n, \ldots, 1/n) = \ED(u_A)$ for $|A| \geq 2$ gives $\eps_a(\ESD) = 1$ for $a \geq 2$.
 In stratified coordinates,
\[
\Phi_{\mathrm{strat}}(\ESD) = (0, 1, 1, \ldots, 1) \in \R^{n-1}.
\]
By Theorem~\ref{thm-stratified-proj}, $\eps^*(\ESD) = \sum_{a \geq 2} w_a = 1 - w_1 = 1 - (n-1)/D_n$, 
recovering Proposition~\ref{prop-eps-ESD}.
 By Theorem~\ref{thm-R2-variance}, since $\eps_a^2 = \eps_a$ for these values,
\[
R^2(\ESD) = \frac{(1 - w_1)^2}{1 - w_1} = 1 - w_1 = \frac{D_n - (n-1)}{D_n},
\]
recovering Proposition~\ref{prop-numerics-n4} and the asymptotic statement of Theorem~\ref{thm-asym-ESD} as a one-line corollary.
\end{example}

\begin{example}[Solidarity Value]\label{ex-So-stratified}
The solidarity value of Nowak and Radzik is known to be efficient, symmetric, and linear, 
and hence lies in $\mL^{\mathrm{ESL}}$.
 Its stratified epsilons admit the closed form
\[
\eps_a(\So) = a \sum_{s=a+1}^{n} \frac{\binom{n-a-1}{s-a-1}}{s \cdot \binom{n-1}{s-1}} \qquad (1 \leq a \leq n-1).
\]
The proof of this in fact follows from the unanimity-game evaluation $\So_i(u_A) = a \sum_{s} \binom{\cdot}{\cdot} \cdot p(s, n)/s$
 (where the binomial coefficient depends on whether $i \in A$, and $p(s,n) = (n-s)!(s-1)!/n!$ is the Shapley weight), 
plus the fact that $\eps_a(\Psi) = n \cdot \beta_a(\Psi)$ for any $\Psi \in \mL^{\mathrm{ESL}}$ since efficiency forces 
$\eps_a$ to be the unique scalar making $(\Psi - \Sh)^{(a)} = \eps_a (\ED - \Sh)^{(a)}$, whose $\beta$-component is $\eps_a/n$.
 (The details are in Appendix~\ref{app-So-formula}.)

For $n = 4$,
\[
\Phi_{\mathrm{strat}}(\So) = \left(\tfrac{23}{36}, \tfrac{13}{18}, \tfrac{3}{4}\right),
\]
giving $\E_w[\eps_a] = 39/58 = \eps^*(\So)$ and $\Var_w(\eps_a) = 19/10092$.
From this it follows immediately that 
 $R^2(\So) = (39/58)^2 / (79/174) = 4563/4582$, recovering Proposition~\ref{prop-So-projection}.

For each fixed $a$, $\eps_a(\So) \to 1$ as $n \to \infty$ (Lemma~\ref{lem-So-eps-asymp}).
Combined with Theorem~\ref{thm-R2-variance}, this yields
\begin{equation}\label{eq-So-asymp}
R^2(\So) \to 1 \qquad \text{as } n \to \infty.
\end{equation}
The convergence is rapid. 
Indeed, 
numerical computation gives $1 - R^2(\So) \approx 4 \cdot 10^{-3}$ at $n = 4$, $4 \cdot 10^{-4}$ at $n = 12$, $9 \cdot 10^{-6}$ at $n = 20$. 
Asymptotically, $\Var_w(\eps_a(\So))$ vanishes because $\eps_a(\So) \to 1$ pointwise in $a$, while $\E_w[\eps_a^2] \to 1$ as well, 
so the relative variance vanishes and $R^2(\So) \to 1$ by Theorem~\ref{thm-R2-variance}.
\end{example}

\subsection{Extension to Nonefficient Targets: The Banzhaf Value}\label{subsec-Bz-stratified}

The Banzhaf value $\Bz$ is symmetric and linear but \emph{not} efficient, 
so it lies in $\mL^{\mathrm{S}}$ but not in $\mL^{\mathrm{ESL}}$. 
The stratified-flat parametrization extends to the larger space $\mL^{\mathrm{S}}$ by enlarging the codomain to capture the efficiency defect.

Fix $\Psi \in \mL^{\mathrm{S}}$. 
By Lemma~\ref{lem-eff-unif-orthogonal}, for each $a < n$, 
the difference $(\Psi - \Sh)^{(a)} \in \mL_a^{\mathrm{S}}$ decomposes uniquely as
\begin{equation}\label{eq-eps-delta-decomp}
(\Psi - \Sh)^{(a)} = \eps_a(\Psi) \cdot (\ED - \Sh)^{(a)} + \delta_a(\Psi) \cdot U^{(a)},
\end{equation}
where $U^{(a)}$ is the unit vector in $\mL_a^{\mathrm{unif}}$ given by $\alpha_a(U^{(a)}) = \beta_a(U^{(a)}) = 1$, and $\delta_a(\Psi) \in \R$ is the \emph{efficiency defect} of $\Psi$ at stratum $a$:
\[
\delta_a(\Psi) = \frac{a \, \alpha_a(\Psi) + (n-a)\, \beta_a(\Psi) - 1}{n} = \frac{1}{n}\left(\sum_{i \in N} \Psi_i(u_A) - 1\right).
\]
For an efficient $\Psi$, $\delta_a(\Psi) = 0$ for every $a$, 
recovering Theorem~\ref{thm-stratified-flat}.

\begin{proposition}[Banzhaf in Stratified Coordinates]\label{prop-Bz-stratified}
For every $a \in \{1, \ldots, n-1\}$,
\[
\eps_a(\Bz) = 1 - \frac{a}{2^{a-1}}, \qquad \delta_a(\Bz) = -\frac{1}{n}\left(1 - \frac{a}{2^{a-1}}\right) = -\frac{\eps_a(\Bz)}{n}.
\]
\end{proposition}

\begin{proof}
On unanimity games, $\alpha_a(\Bz) = 1/2^{a-1}$ and $\beta_a(\Bz) = 0$. 
Then $\eps_a(\Bz)$ is the coefficient of $(\Psi - \Sh)^{(a)}$ in the $(\ED - \Sh)^{(a)}$ direction; 
computing the inner product against $(\ED - \Sh)^{(a)}$ in stratum-$a$ coordinates,
\[
\eps_a(\Bz) = \frac{a \, (\alpha_a(\Bz) - \alpha_a(\Sh))(\alpha_a(\ED) - \alpha_a(\Sh)) + (n-a)\, (\beta_a(\Bz) - 0)(\beta_a(\ED))}{a(\alpha_a(\ED) - \alpha_a(\Sh))^2 + (n-a)\beta_a(\ED)^2}.
\]
Substituting $\alpha_a(\Bz) - \alpha_a(\Sh) = 1/2^{a-1} - 1/a$, $\beta_a(\Bz) = 0$, $\alpha_a(\ED) - \alpha_a(\Sh) = 1/n - 1/a$, $\beta_a(\ED) = 1/n$,
 and simplifying via the algebraic identity $(1/n - 1/a) = -(n-a)/(na)$,
\[
\eps_a(\Bz) = \frac{a \cdot (1/2^{a-1} - 1/a) \cdot (-(n-a)/(na))}{(n-a)/(na)} = -a \cdot (1/2^{a-1} - 1/a) = 1 - \frac{a}{2^{a-1}}.
\]
For $\delta_a(\Bz)$, $a \alpha_a(\Bz) + (n-a) \beta_a(\Bz) = a/2^{a-1}$, so $\delta_a(\Bz) = (a/2^{a-1} - 1)/n = -\eps_a(\Bz)/n$.
\end{proof}

\begin{theorem}[Generalized Pythagorean Decomposition]\label{thm-gen-pythag}
For every $\Psi \in \mL^{\mathrm{S}}$,
\[
\|\Psi - \Sh\|^2_\mL = \sum_{a=1}^{n-1} \binom{n}{a}\big(\tfrac{1}{a} - \tfrac{1}{n}\big) \eps_a(\Psi)^2 + \sum_{a=1}^{n-1} \binom{n}{a} \cdot n \cdot \delta_a(\Psi)^2 + \|(\Psi - \Sh)^{(n)}\|^2_\mL,
\]
and
\[
R^2(\Psi) = \frac{\big(\sum_a w_a \eps_a(\Psi)\big)^2 \cdot D_n}{\|\Psi - \Sh\|^2_\mL}.
\]
\end{theorem}

\begin{proof}
By Lemmas~\ref{lem-stratum-orthogonal} and~\ref{lem-eff-unif-orthogonal}, 
$\|\Psi - \Sh\|^2_\mL = \sum_a \|(\Psi - \Sh)^{(a)}\|^2_\mL$, and within each stratum $a < n$ the eff/unif 
decomposition contributes $\eps_a(\Psi)^2 \|(\ED - \Sh)^{(a)}\|^2_\mL + \delta_a(\Psi)^2 \|U^{(a)}\|^2_\mL$. 
The norms are $\|(\ED - \Sh)^{(a)}\|^2_\mL = \binom{n}{a}(1/a - 1/n)$ (Section~\ref{sec-best-approx}) 
and $\|U^{(a)}\|^2_\mL = \binom{n}{a} \cdot n$ (from \eqref{eq-stratum-IP} with $\alpha_a = \beta_a = 1$). 
Combining gives the squared-norm formula. 
The $R^2$ formula then follows from $\langle \Psi - \Sh, \ED - \Sh \rangle_\mL = D_n \sum_a w_a \eps_a(\Psi)$ (Theorem~\ref{thm-stratified-proj}, 
whose proof in fact  did not  use efficiency of $\Psi$).
\end{proof}

\begin{corollary}[Stratified Picture of $R^2(\Bz) \to 1/2$]\label{cor-Bz-asym-stratified}
As $n \to \infty$, $R^2(\Bz) \to 1/2$.
\end{corollary}

\begin{proof}
By Theorem~\ref{thm-gen-pythag} and Proposition~\ref{prop-Bz-stratified},
 using $\delta_a(\Bz) = -\eps_a(\Bz)/n$ and the algebraic identity $(1/a - 1/n) + 1/n = 1/a$,
\begin{align}
\|\Bz - \Sh\|^2_\mL 
&= \sum_{a=1}^{n-1}\!\binom{n}{a}\!\left(\tfrac{1}{a} - \tfrac{1}{n}\right)\!\eps_a(\Bz)^2 + \sum_{a=1}^{n-1}\!\binom{n}{a} n \cdot \frac{\eps_a(\Bz)^2}{n^2} + \|(\Bz-\Sh)^{(n)}\|^2_\mL \notag \\
&= \sum_{a=1}^{n-1} \frac{\binom{n}{a}}{a}\, \eps_a(\Bz)^2 + \|(\Bz-\Sh)^{(n)}\|^2_\mL. \label{eq-Bz-norm-clean}
\end{align}
Since $1 - \eps_a(\Bz) = a/2^{a-1}$ (Proposition~\ref{prop-Bz-stratified}) and $1 - \eps_a^2 \leq 2(1 - \eps_a)$,
\[
0 \;\leq\; \sum_{a=1}^{n-1} \frac{\binom{n}{a}}{a}\bigl(1 - \eps_a(\Bz)^2\bigr) \;\leq\; 2 \sum_{a=1}^{n-1} \frac{\binom{n}{a}}{2^{a-1}} \;\leq\; 4\!\left[\!\left(\tfrac{3}{2}\right)^{\!n} - 1\right]\!.
\]
Combined with $\sum_{a=1}^{n-1} \binom{n}{a}/a = H_n - 1/n$ and $\|(\Bz-\Sh)^{(n)}\|^2_\mL = n(1/2^{n-1} - 1/n)^2 = O(1/n)$, equation~\eqref{eq-Bz-norm-clean} gives
\[
\|\Bz - \Sh\|^2_\mL = H_n - O\!\bigl((3/2)^n\bigr) = H_n\bigl(1 + o(1)\bigr).
\]
By Lemma~\ref{lem-Hn-formula}, $H_n = \Sigma_n + 2^n/n - H_n^{\mathrm{harm}}$ where $\Sigma_n := \sum_{j=1}^{n-1} 2^j/j$. 
The proof of Lemma~\ref{lem-Dn-asymptotic} establishes $\Sigma_n = (2^n/n)(1 + O(1/n))$, 
and $H_n^{\mathrm{harm}} = O(\log n)$ is dominated by $2^n/n$. 
Therefore,  $H_n = (2^{n+1}/n)(1 + O(1/n))$. 
Combined with $D_n = (2^n/n)(1 + O(1/n))$ from Lemma~\ref{lem-Dn-asymptotic}, 
$H_n/D_n \to 2$ and hence $\|\Bz - \Sh\|^2_\mL = 2 D_n(1 + o(1))$. 
By Lemma~\ref{lem-eps-Bz-asymptotic}, $\eps^*(\Bz) \to 1$, 
so it holds that  $(\eps^*(\Bz))^2 \cdot D_n = D_n(1 + o(1))$.
 Dividing,
\[
R^2(\Bz) = \frac{(\eps^*(\Bz))^2 \cdot D_n}{\|\Bz - \Sh\|^2_\mL} = \frac{D_n(1+o(1))}{2 D_n(1 + o(1))} \longrightarrow \tfrac{1}{2}. \qedhere
\]
\end{proof}

The asymptotic half-alignment of the Banzhaf value established in Corollary~\ref{cor-Bz-asym-stratified} now has a transparent structural cause: 
the identity $\delta_a(\Bz) = -\eps_a(\Bz)/n$ of Proposition~\ref{prop-Bz-stratified} locks the efficiency defect to the egalitarian-Shapley deviation in a $1{:}1$ ratio at every stratum. 
Summed across strata, this yields $\|\Bz - \Sh\|^2 \sim H_n \sim 2 D_n$ against a projection contribution of only $D_n$,
 leaving exactly half of the Banzhaf value's deviation from $\Sh$ permanently outside the egalitarian Shapley axis.

\subsection{The Three-Way Classification, Structurally}\label{subsec-classification-structural}

Combining Examples~\ref{ex-ESD-stratified}, \ref{ex-So-stratified},
 and Proposition~\ref{prop-Bz-stratified}, the three-way classification of Section~\ref{subsec-classification}
 acquires a sharp structural interpretation:

\begin{itemize}[leftmargin=2em]
\item \emph{Equal-surplus-division}: $\Phi_{\mathrm{strat}}(\ESD) = (0, 1, \ldots, 1)$.
  This is almost-diagonal, deviating from the egalitarian Shapley line only at the singleton stratum, and the deviation vanishes asymptotically as $w_1 \to 0$. 
  Hence, $R^2(\ESD) \to 1$.
\item \emph{Solidarity}: $\Phi_{\mathrm{strat}}(\So) = (\eps_1(\So), \ldots, \eps_{n-1}(\So))$ with $\eps_a(\So) \to 1$ for every fixed $a$, but the entries are \emph{distinct} at finite $n$. 
 The deviation $1 - R^2(\So)$ is the relative weighted variance of the $\eps_a(\So)$, 
 which is small but nonzero, approaching zero as $n \to \infty$.
\item \emph{Banzhaf}: the Banzhaf value lies in $\mL^{\mathrm{S}}$ but not in $\mL^{\mathrm{ESL}}$, since it fails efficiency. 
The efficiency defect satisfies $\delta_a(\Bz) = -\eps_a(\Bz)/n$ at every stratum, which forces $R^2(\Bz) \to 1/2$.
\end{itemize}

The classification is thus not a numerical coincidence at $n = 4$: 
each of $\ESD$, $\So$, $\Bz$ occupies a structurally distinct location in $\mL^{\mathrm{S}}$, distinguished by which type of deviation from $\Sh$ each one displays. 
The three types are: a single off-diagonal stratum coordinate ($\ESD$), a smooth profile of stratified epsilons ($\So$), or a nontrivial efficiency defect ($\Bz$).

\section{Asymptotic Behavior as $n \to \infty$}\label{sec-asymptotic}

The numerical results of Section~\ref{sec-pythagorean} were established at $n = 4$.
 We now investigate how the optimal parameter $\eps^*(\Psi)$ and the goodness-of-fit $R^2(\Psi)$ behave as the number of players grows.
  The main results of this section establish $R^2(\ESD) \to 1$ (perfect asymptotic alignment with the egalitarian Shapley line) and $R^2(\Bz) \to 1/2$ (a permanent half-and-half geometric split).
   The corresponding result for the solidarity value, $R^2(\So) \to 1$, was already obtained in Section~\ref{sec-spectral} via the spectral decomposition (Equation~\eqref{eq-So-asymp});
    together the three results form the asymptotic three-way classification.

\subsection{Asymptotic Estimate of $D_n$}

We start with the denominator $D_n = \langle \ED - \Sh, \ED - \Sh \rangle_{\mL}$ that appears throughout.

\begin{lemma}\label{lem-Hn-formula}
For all $n \geq 1$,
\[
H_n = \sum_{a=1}^n \frac{\binom{n}{a}}{a} = \sum_{j=1}^n \frac{2^j}{j} - \sum_{j=1}^n \frac{1}{j}.
\]
\end{lemma}

\begin{proof}
Using $\int_0^1 t^{a-1} \, dt = 1/a$,
\[
H_n = \sum_{a=1}^n \binom{n}{a} \int_0^1 t^{a-1} \, dt = \int_0^1 \frac{(1+t)^n - 1}{t} \, dt.
\]
Substituting $u = 1 + t$ and factoring $u^n - 1 = (u-1) \sum_{k=0}^{n-1} u^k$,
\[
\int_0^1 \frac{(1+t)^n - 1}{t} \, dt = \int_1^2 \sum_{k=0}^{n-1} u^k \, du = \sum_{k=0}^{n-1} \frac{2^{k+1} - 1}{k+1} = \sum_{j=1}^n \frac{2^j - 1}{j} = \sum_{j=1}^n \frac{2^j}{j} - \sum_{j=1}^n \frac{1}{j}. \qedhere
\]
\end{proof}

\begin{lemma}\label{lem-Dn-asymptotic}
As $n \to \infty$,
\[
D_n = \frac{2^n}{n} \cdot \left(1 + O\!\left(\frac{1}{n}\right)\right).
\]
\end{lemma}

\begin{proof}
Recall that  $D_n = H_n - (2^n - 1)/n$ from Section~\ref{sec-best-approx}.
 By Lemma~\ref{lem-Hn-formula}, $H_n = \sum_{j=1}^n 2^j/j - H_n^{\mathrm{harm}}$ where $H_n^{\mathrm{harm}} := \sum_{j=1}^n 1/j = O(\log n)$. 
Separating off  the $j = n$ term and combining it with $-(2^n - 1)/n$ yields the exact identity
\[
D_n = \sum_{j=1}^{n-1} \frac{2^j}{j} + \frac{1}{n} - H_n^{\mathrm{harm}}.
\]
It therefore suffices to estimate $\Sigma_n := \sum_{j=1}^{n-1} 2^j/j$. 
Factoring out the largest term,
\[
\Sigma_n = \frac{2^{n-1}}{n-1} \cdot B_n, \qquad B_n := \sum_{k=0}^{n-2} \frac{n-1}{(n-1-k) \cdot 2^k},
\]
and we claim $B_n = 2 + O(1/n)$.
Write $B_n = G_n + R_n$, where $G_n := \sum_{k=0}^{n-2} 1/2^k$ is the geometric sum while 
the quantity $R_n := \sum_{k=0}^{n-2} k / ((n-1-k) \cdot 2^k)$ collects the corrections. 
Clearly it holds that  $G_n = 2 - 2^{2-n} = 2 + O(2^{-n})$.
For $R_n$, split the range at $k_0 := \lfloor n/2 \rfloor - 1$. 
For $k \leq k_0$, 
we have $n - 1 - k \geq n/2$, so each term is at most $2k/(n \cdot 2^k)$, and
\[
\sum_{k=0}^{k_0} \frac{k}{(n-1-k) \cdot 2^k} \leq \frac{2}{n} \sum_{k=0}^{\infty} \frac{k}{2^k} = \frac{4}{n}.
\]
For $k > k_0$, we have $1/2^k \leq 2^{1 - n/2}$ and $k/(n-1-k) \leq k \leq n$, so the tail contributes at most $(n/2) \cdot n \cdot 2^{1 - n/2} = O(n^2 / 2^{n/2})$, which is $o(1/n)$. 
Hence, $R_n = O(1/n)$ and $B_n = 2 + O(1/n)$.
Therefore,
\[
\Sigma_n = \frac{2^{n-1}}{n-1}\bigl(2 + O(1/n)\bigr) = \frac{2^n}{n-1}\bigl(1 + O(1/n)\bigr) = \frac{2^n}{n}\bigl(1 + O(1/n)\bigr),
\]
where the last step uses $n/(n-1) = 1 + 1/(n-1) = 1 + O(1/n)$. 
Substituting back,
\[
D_n = \Sigma_n + \frac{1}{n} - H_n^{\mathrm{harm}} = \frac{2^n}{n}\bigl(1 + O(1/n)\bigr) + O(\log n) = \frac{2^n}{n}\bigl(1 + O(1/n)\bigr),
\]
since $\log n = o(2^n / n^2)$.
\end{proof}

\subsection{Asymptotic Alignment of $\ESD$}

\begin{theorem}\label{thm-asym-ESD}
For every $n \geq 3$,
\[
1 - R^2(\ESD) = \frac{n - 1}{D_n},
\]
and for $n=2$ it holds that $\ESD=\Sh$ trivially.
In particular, as $n \to \infty$,
\[
\eps^*(\ESD) \to 1 \quad \text{with} \quad 1 - \eps^*(\ESD) = O\!\left(\frac{n^2}{2^n}\right),
\]
and
\[
R^2(\ESD) \to 1 \quad \text{with} \quad 1 - R^2(\ESD) = O\!\left(\frac{n^2}{2^n}\right).
\]
\end{theorem}

\begin{proof}
Recall from Lemma~\ref{lem-norm-ESD} that the squared norm equals $\|\ESD - \Sh\|^2_{\mL} = D_n - (n-1)$, 
and from Proposition~\ref{prop-eps-ESD} that the optimal parameter equals $\eps^*(\ESD) = 1 - (n - 1)/D_n$.
Hence,
\begin{align*}
(\eps^*(\ESD))^2 \cdot D_n &= \left(1 - \tfrac{n-1}{D_n}\right)^2 \cdot D_n = D_n - 2(n-1) + \frac{(n-1)^2}{D_n}, \\
\|\ESD - \Sh\|^2_{\mL} &= D_n - (n - 1).
\end{align*}
We conclude that 
\begin{align*}
1 - R^2(\ESD)
&= \frac{\|\ESD - \Sh\|^2_{\mL} - (\eps^*(\ESD))^2 \cdot D_n}{\|\ESD - \Sh\|^2_{\mL}} \\
&= \frac{(D_n - (n-1)) - (D_n - 2(n-1) + (n-1)^2/D_n)}{D_n - (n-1)} \\
&= \frac{(n-1) - (n-1)^2/D_n}{D_n - (n-1)} \\
&= \frac{(n-1)\bigl(D_n - (n-1)\bigr)/D_n}{D_n - (n-1)} \\
&= \frac{n-1}{D_n},
\end{align*}
which proves the exact identity. 
The asymptotic statements follow from Lemma~\ref{lem-Dn-asymptotic}: 
since $D_n = (2^n/n)(1 + O(1/n))$, we have $(n-1)/D_n = (n-1) \cdot n / 2^n \cdot (1 + O(1/n)) = O(n^2 / 2^n)$, 
which gives both the rate for $1 - \eps^*(\ESD)$ and for $1 - R^2(\ESD)$.
\end{proof}

\subsection{Asymptotic Half-Alignment of $\Bz$}

The Banzhaf value behaves differently. 
To analyze it we need an asymptotic expression for $\|\Bz - \Sh\|^2_{\mL}$.

\begin{lemma}\label{lem-Bz-norm-asymptotic}
As $n \to \infty$,
\[
\|\Bz - \Sh\|^2_{\mL} = \frac{2^{n+1}}{n} \cdot \left(1 + O\!\left(\frac{1}{n}\right)\right).
\]
\end{lemma}

\begin{proof}
By Lemma~\ref{lem-norm-Bz},
\[
\|\Bz - \Sh\|^2_{\mL} = \sum_{a=1}^n \binom{n}{a} \cdot a \left(\frac{1}{2^{a-1}} - \frac{1}{a}\right)^{\!2}.
\]
Expanding the square and multiplying through by $a$,
\[
a \left(\frac{1}{2^{a-1}} - \frac{1}{a}\right)^{\!2} = \frac{a}{4^{a-1}} - \frac{2}{2^{a-1}} + \frac{1}{a}.
\]
Hence, $\|\Bz - \Sh\|^2_{\mL} = T_1 - 2 T_2 + T_3$, where
\[
T_1 = \sum_{a=1}^n \binom{n}{a} \frac{a}{4^{a-1}}, \qquad T_2 = \sum_{a=1}^n \binom{n}{a} \frac{1}{2^{a-1}}, \qquad T_3 = \sum_{a=1}^n \frac{\binom{n}{a}}{a} = H_n.
\]

\emph{Computing $T_1$.} 
Using $a \binom{n}{a} = n \binom{n-1}{a-1}$,
\[
T_1 = 4n \sum_{a=1}^n \binom{n-1}{a-1} \frac{1}{4^a} = n \sum_{k=0}^{n-1} \binom{n-1}{k} \frac{1}{4^k} = n \cdot \left(\frac{5}{4}\right)^{\!n-1}.
\]

\emph{Computing $T_2$.} 
Directly, $T_2 = 2 \sum_{a=1}^n \binom{n}{a}/2^a = 2((3/2)^n - 1)$.

\emph{Asymptotic comparison.} 
By proof of Lemma~\ref{lem-Dn-asymptotic},
$\Sigma_n := \sum_{j=1}^{n-1} 2^j/j = (2^n/n)(1 + O(1/n))$.
 We therefore have 
\[
T_3 = H_n = \frac{2^n}{n} + \Sigma_n - H_n^{\mathrm{harm}} = \frac{2^{n+1}}{n}\bigl(1 + O(1/n)\bigr),
\]
since the harmonic correction $H_n^{\mathrm{harm}} = O(\log n)$ is dominated by $2^n / n^2$.
For $T_1$ and $T_2$ we compare to $T_3$ directly:
\[
\frac{T_1}{T_3} = \frac{n(5/4)^{n-1} \cdot n}{2^{n+1}} \cdot (1 + o(1)) = \frac{2 n^2}{5} \cdot \left(\frac{5}{8}\right)^{\!n} (1 + o(1)) = O\bigl(n^2 (5/8)^n\bigr),
\]
\[
\frac{T_2}{T_3} = \frac{2\bigl((3/2)^n - 1\bigr) \cdot n}{2^{n+1}} \cdot (1 + o(1)) = n \cdot \left(\frac{3}{4}\right)^{\!n} (1 + o(1)) = O\bigl(n (3/4)^n\bigr).
\]
Both ratios tend to $0$.
  Therefore,
\[
\|\Bz - \Sh\|^2_{\mL} = T_3 \bigl(1 - 2T_2/T_3 + T_1/T_3\bigr) = T_3 \bigl(1 + o(1)\bigr) = \frac{2^{n+1}}{n} \cdot \bigl(1 + O(1/n)\bigr),
\]
where the final $O(1/n)$ comes from the $T_3$ asymptotic itself,
   the $T_1$ and $T_2$ corrections being of strictly lower order (indeed, exponentially smaller).
\end{proof}

\begin{lemma}\label{lem-eps-Bz-asymptotic}
As $n \to \infty$, $\eps^*(\Bz) \to 1$ and
\[
1 - \eps^*(\Bz) = \frac{4n}{3}\left(\frac{3}{4}\right)^{\!n}\bigl(1 + O(1/n)\bigr) = O\bigl(n (3/4)^n\bigr).
\]
\end{lemma}

\begin{proof}
By Proposition~\ref{prop-eps-banzhaf},
\[
1 - \eps^*(\Bz) = \frac{2 \cdot (3/2)^{n-1} - 2}{D_n} = \frac{(4/3) \cdot (3/2)^n - 2}{D_n}.
\]
The numerator equals $(4/3)(3/2)^n - 2 = (4/3)(3/2)^n \bigl(1 - 2 \cdot (3/4) \cdot (2/3)^n\bigr) = (4/3)(3/2)^n (1 + o(1))$, 
since $(2/3)^n \to 0$.
 By Lemma~\ref{lem-Dn-asymptotic}, $D_n = (2^n / n)(1 + O(1/n))$. 
Therefore,
\[
1 - \eps^*(\Bz) = \frac{(4/3)(3/2)^n (1 + o(1))}{(2^n/n)(1 + O(1/n))} = \frac{4n}{3} \cdot \left(\frac{3}{4}\right)^{\!n} \bigl(1 + O(1/n)\bigr),
\]
which is $O(n(3/4)^n) \to 0$.
\end{proof}

\begin{theorem}\label{thm-asym-Bz}
As $n \to \infty$,
\[
R^2(\Bz) = \frac{1}{2}\bigl(1 + O(1/n)\bigr) \to \frac{1}{2}.
\]
\end{theorem}

\begin{proof}
By Lemma~\ref{lem-eps-Bz-asymptotic}, $1 - \eps^*(\Bz) = O(n(3/4)^n)$, so
\[
(\eps^*(\Bz))^2 = 1 - 2(1 - \eps^*(\Bz)) + (1 - \eps^*(\Bz))^2 = 1 + O(n(3/4)^n).
\]
By Lemma~\ref{lem-Dn-asymptotic}, $D_n = (2^n/n)(1 + O(1/n))$. 
Hence, the numerator of $R^2(\Bz)$ is
\[
(\eps^*(\Bz))^2 \cdot D_n = \frac{2^n}{n} \cdot \bigl(1 + O(1/n)\bigr).
\]
By Lemma~\ref{lem-Bz-norm-asymptotic}, the denominator is
\[
\|\Bz - \Sh\|^2_{\mL} = \frac{2^{n+1}}{n} \cdot \bigl(1 + O(1/n)\bigr).
\]
Dividing,
\[
R^2(\Bz) = \frac{(2^n/n)(1 + O(1/n))}{(2^{n+1}/n)(1 + O(1/n))} = \frac{1}{2} \cdot \bigl(1 + O(1/n)\bigr) \to \frac{1}{2}. \qedhere
\]
\end{proof}

\subsection{Discussion of the Asymptotic Results}

Theorems~\ref{thm-asym-ESD} and~\ref{thm-asym-Bz} together exhibit a striking qualitative contrast.

The equal-surplus-division value becomes \emph{asymptotically aligned} with the egalitarian Shapley line: 
as $n$ grows, both $\eps^*(\ESD)$ and $R^2(\ESD)$ converge to $1$ at exponential rate. 
Geometrically, $\ESD$ is, for large $n$, almost identical to the equal-division solution $\ED$ in our inner-product, 
and the line $\{\Feps\}$ captures essentially all of its variation around $\Sh$. 
The intuition is structural: $\ESD$ differs from $\ED$ only on singleton-coalition values, 
and singletons constitute an asymptotically vanishing fraction of all coalitions.

The Banzhaf value, however, behaves very differently. 
While $\eps^*(\Bz) \to 1$ as well, the goodness-of-fit $R^2(\Bz)$ converges not to $1$ but to $1/2$. 
 This is the key qualitative finding: the projection of $\Bz$ onto the egalitarian Shapley line captures, asymptotically, 
 exactly half of the squared distance from $\Sh$, with the other half permanently in the orthogonal residual. 
  The Banzhaf value therefore retains a robust geometric distinctness from the egalitarian Shapley family at every scale:
    it is neither aligned with the line (as $\ESD$ becomes) nor orthogonal to it (as it nearly is at small $n$). 
   It occupies a half-and-half geometric position that no amount of growth in $n$ can change.

We summarize the contrast:

\smallskip

\begin{center}
\renewcommand{\arraystretch}{1.4}
\begin{tabular}{lcc}
\hline
$\Psi$ & $\lim_{n \to \infty} \eps^*(\Psi)$ & $\lim_{n \to \infty} R^2(\Psi)$ \\
\hline
$\Sh$  & $0$ & $1$ (conv.) \\
$\ED$  & $1$ & $1$ \\
$\Bz$  & $1$ & $1/2$ \\
$\ESD$ & $1$ & $1$ \\
$\So$  & $1$ & $1$ \\
\hline
\end{tabular}
\end{center}

\smallskip

 The asymptotic result $R^2(\Bz) \to 1/2$ established above admits the following sharp and somewhat surprising interpretation: 
 the distinctness of the Banzhaf value from the egalitarian Shapley family is by no means a small-$n$ artifact that gradually washes out as the player set grows ever larger.
  It is structurally permanent. 
  Half of Banzhaf's deviation from Shapley lives, asymptotically and forever, in directions that no convex combination of $\Sh$ and $\ED$ can reach.
    This makes Banzhaf qualitatively different from $\ESD$, whose distinctness from the egalitarian Shapley line vanishes as $n \to \infty$. 
    The geometric classification of Section~\ref{sec-pythagorean}, established at $n = 4$, is therefore not a coincidence of small player sets but a stable structural feature.

\section{Multi-Parameter Best Fit}\label{sec-multiparam}

The framework of Section~\ref{sec-best-approx} approximates a target $\Psi$ by the best member of a one-parameter family. 
We now extend it to richer affine subspaces of $\mL$, 
providing finer geometric resolution of where any given value map sits.

\subsection{Best Two-Parameter Fit through $\Sh$, $\ED$, and $\Bz$}

For $n \geq 3$, consider the affine $2$-plane through $\Sh$ in the directions $\ED - \Sh$ and $\Bz - \Sh$:
\[
\mathcal{S} := \Sh + \mathrm{span}_{\R}(\ED - \Sh, \Bz - \Sh) \subset \mL.
\]
Each element of $\mathcal{S}$ has the form
\[
F^{\eps,\delta} := (1 - \eps - \delta) \Sh + \eps \cdot \ED + \delta \cdot \Bz.
\]
This is a two-parameter family that interpolates between three canonical solutions, 
satisfying efficiency, symmetry and linearity for all $(\eps, \delta) \in \R^2$.

\begin{lemma}\label{lem-independence}
For $n \geq 3$, the maps $\ED - \Sh$ and $\Bz - \Sh$ are linearly independent in $\mL$.
\end{lemma}

\begin{proof}
Suppose $\alpha(\ED - \Sh) + \beta(\Bz - \Sh) = 0$.
 Evaluating at $u_N$ we obtain the following: $\ED(u_N) = \Sh(u_N) = (1/n, \ldots, 1/n)$, so $(\ED - \Sh)(u_N) = 0$. 
On the other hand, we have 
$\Bz(u_N) = (1/2^{n-1}, \ldots, 1/2^{n-1})$, so $(\Bz - \Sh)(u_N)_i = 1/2^{n-1} - 1/n \neq 0$ for $n \geq 3$ (since $1/2^{n-1} \neq 1/n$). 
Hence, $\beta = 0$. 
Then $\alpha(\ED - \Sh) = 0$; evaluating at $u_{N \setminus \{1\}}$, 
the result is nonzero, forcing $\alpha = 0$.
\end{proof}

\begin{theorem}[Best Two-Parameter Fit]\label{thm-best-2param}
Let $n \geq 3$ and $\Psi \in \mL$. 
There exists a unique pair $(\eps^{**}(\Psi), \delta^{**}(\Psi)) \in \R^2$ minimizing $\|F^{\eps,\delta} - \Psi\|_{\mL}$. 
It satisfies the normal equations
\[
\begin{pmatrix}
\langle \ED - \Sh, \ED - \Sh \rangle_{\mL} & \langle \ED - \Sh, \Bz - \Sh \rangle_{\mL} \\
\langle \Bz - \Sh, \ED - \Sh \rangle_{\mL} & \langle \Bz - \Sh, \Bz - \Sh \rangle_{\mL}
\end{pmatrix}
\begin{pmatrix} \eps^{**}(\Psi) \\ \delta^{**}(\Psi) \end{pmatrix}
=
\begin{pmatrix}
\langle \ED - \Sh, \Psi - \Sh \rangle_{\mL} \\
\langle \Bz - \Sh, \Psi - \Sh \rangle_{\mL}
\end{pmatrix}.
\]
\end{theorem}

\begin{proof}
This is the standard normal-equation result for orthogonal projection of $\Psi - \Sh$ onto the $2$-dimensional subspace $\mathrm{span}(\ED - \Sh, \Bz - \Sh)$. 
Existence and uniqueness follow from Lemma~\ref{lem-independence}, 
which guarantees that the Gram-matrix is invertible.
\end{proof}

\subsection{Worked Example: $\Psi = \ESD$ at $n = 4$}

We compute the best two-parameter fit for the equal-surplus-division value at $n = 4$.

The Gram-matrix entries are:
\begin{itemize}[leftmargin=2em]
\item $\langle \ED - \Sh, \ED - \Sh \rangle_{\mL} = D_4 = 29/6$ (Lemma~\ref{lem-Dn-values}).
\item $\langle \Bz - \Sh, \Bz - \Sh \rangle_{\mL} = 7/48$ (Proposition~\ref{prop-numerics-n4}).
\item $\langle \ED - \Sh, \Bz - \Sh \rangle_{\mL} 
= D_4 + 2 - 2(3/2)^3 = 1/12$
 (Proposition~\ref{prop-eps-banzhaf}).
\end{itemize}

The right-hand side entries for $\Psi = \ESD$ are:
\begin{itemize}[leftmargin=2em]
\item $\langle \ED - \Sh, \ESD - \Sh \rangle_{\mL} = D_4 - (n-1) = 29/6 - 3 = 11/6$ (Proposition~\ref{prop-eps-ESD}).
\item $\langle \Bz - \Sh, \ESD - \Sh \rangle_{\mL} = 1/12$. 
\emph{Computation:} $(\ESD - \Sh)(u_A)$ vanishes for $|A| = 1$ and equals $(\ED - \Sh)(u_A)$ for $|A| \geq 2$. 
The pairing with $(\Bz - \Sh)(u_A)$ vanishes for $|A| \in \{1, 2\}$ (the inner factor $1/2^{a-1} - 1/a$ vanishes at $a = 1, 2$). 
For $a = 3$: 
$\binom{4}{3} \cdot 3 \cdot (1/4 - 1/3)(1/4 - 1/3) = 1/12$. 
For $a = 4$:
 the factor $1/4 - 1/4 = 0$ kills the contribution. 
Total: $1/12$.
\end{itemize}

The Gram-matrix and its determinant are
\[
G = \begin{pmatrix} 29/6 & 1/12 \\ 1/12 & 7/48 \end{pmatrix}, \qquad
\det G 
= \frac{29}{6} \cdot \frac{7}{48} - \frac{1}{144} 
= \frac{203}{288} - \frac{2}{288} 
= \frac{67}{96}.
\]
By Cramer's rule,
\begin{align*}
\eps^{**}(\ESD) &= \frac{1}{67/96} \det
 \begin{pmatrix} 
 11/6 & 1/12 \\ 1/12 & 7/48
  \end{pmatrix} 
  = \frac{96}{67} \left( \frac{11}{6} \cdot \frac{7}{48} - \frac{1}{144} \right) = \frac{96}{67} \cdot \frac{75}{288} = \frac{25}{67}, \\[0.5em]
\delta^{**}(\ESD) &= \frac{1}{67/96} \det
 \begin{pmatrix} 29/6 & 11/6 \\ 1/12 & 1/12 
 \end{pmatrix} 
 = \frac{96}{67} \cdot \frac{1}{12}\left( \frac{29}{6} - \frac{11}{6} \right) = \frac{96}{67} \cdot \frac{1}{4} = \frac{24}{67}.
\end{align*}

\begin{proposition}\label{prop-best-2param-ESD}
For $n = 4$, the best approximation of $\ESD$ within $\Sh + \mathrm{span}(\ED - \Sh, \Bz - \Sh)$ is
\[
F^{25/67,\, 24/67} = \frac{18}{67} \Sh + \frac{25}{67} \ED + \frac{24}{67} \Bz.
\]
The associated goodness-of-fit, defined as $R^2_2(\ESD) := \|\text{projection}\|^2_{\mL} / \|\ESD - \Sh\|^2_{\mL}$, is equal 
to  $287/737 \approx 38.94\%$, compared with $R^2(\ESD) = 11/29 \approx 37.93\%$ for the one-parameter fit.
\end{proposition}

\begin{proof}
The form of the projection follows from the computation of $(\eps^{**}, \delta^{**})$ and the identity $F^{\eps,\delta} = (1 - \eps - \delta)\Sh + \eps \ED + \delta \Bz$. 
With $\eps^{**} = 25/67$ and $\delta^{**} = 24/67$, the Shapley coefficient is $1 - 49/67 = 18/67$.
The squared norm of the projection is given by the standard formula
\[
\|\text{proj}\|^2_{\mL} = \eps^{**} \langle \ED - \Sh, \ESD - \Sh \rangle_{\mL} + \delta^{**} \langle \Bz - \Sh, \ESD - \Sh \rangle_{\mL} = \frac{25}{67} \cdot \frac{11}{6} + \frac{24}{67} \cdot \frac{1}{12} 
= \frac  {287}{402}.
\]
By Lemma~\ref{lem-norm-ESD}, $\|\ESD - \Sh\|^2_{\mL} = 11/6 = 737/402$. 
 Therefore, 
 $R^2_2(\ESD) = 287/737$.
By a direct numerical comparison 
we can readily deduce that the gain is equal to the quantity 
 $R^2_2 - R^2 = 287/737 - 11/29 \approx 0.0101$, i.e., approximately $1\%$.
\end{proof}

\begin{remark}\label{rem-2param-ESD}
The result is striking: among $\Sh$, $\ED$ and $\Bz$, 
the equal-surplus-division value sits remarkably close to the centroid, with weights $18/67 \approx 27\%$, $25/67 \approx 37\%$ and $24/67 \approx 36\%$ respectively. 
The ESD value is therefore well captured as a near-balanced mixture of the three canonical solutions. 
The marginal improvement of $R^2$ from the one-parameter fit ($\approx 37.9\%$) to the two-parameter fit ($\approx 38.9\%$) 
is modest because $\Bz$ is itself nearly orthogonal to the egalitarian Shapley axis (Remark~\ref{rem-Bz-orthogonal}); 
adding the Banzhaf direction therefore captures only a small additional component of $\ESD$.
\end{remark}

\subsection{General Pythagorean Identity}

The Pythagorean decomposition extends to the multi-parameter setting. 
Let $U \subseteq \mL$ be any finite-dimensional subspace, 
and let $P_U$ denote the orthogonal projection onto $\Sh + U$ within $\mL$. 
Then for every $\Psi \in \mL$,
\[
\|\Psi - \Sh\|^2_{\mL} = \|P_U(\Psi) - \Sh\|^2_{\mL} + \|\Psi - P_U(\Psi)\|^2_{\mL},
\]
yielding the goodness-of-fit measure
\[
R^2_U(\Psi) := \frac{\|P_U(\Psi) - \Sh\|^2_{\mL}}{\|\Psi - \Sh\|^2_{\mL}} \in [0, 1].
\]
For $U = \mathrm{span}(\ED - \Sh)$, this recovers Definition~\ref{def-Rsquared};
 for $U = \mathrm{span}(\ED - \Sh, \Bz - \Sh)$, 
 it gives the two-parameter version of Proposition~\ref{prop-best-2param-ESD}; 
 and the framework extends to arbitrary finite-dimensional $U \subseteq \mL$ by the same normal-equation construction.

\section{Discussion and Open Directions}\label{sec-discussion}

We have presented an inner-product geometry on linear value maps in cooperative game theory, 
built on a basis-independent inner product on the space $\mL$ of linear value maps. 
The structural core of the framework is the canonical isomorphism $\mL^{\mathrm{ESL}} \cong \R^{n-1}$ of Section~\ref{sec-spectral}: 
every efficient symmetric linear value map decomposes uniquely into $n - 1$ stratified epsilons, 
one per coalition size, with the classical egalitarian Shapley family of Joosten appearing as the one-parameter diagonal slice. 
Within this picture, the orthogonal projection onto the egalitarian Shapley line (Sections~\ref{sec-best-approx} 
and~\ref{sec-pythagorean}) becomes explicit and statistically transparent: 
the optimal parameter $\eps^*(\Psi)$ is the weighted mean of the stratified epsilons under an explicit probability distribution $\{w_a\}$ over coalition sizes, 
and the goodness-of-fit $R^2(\Psi)$ is one minus the relative weighted variance of those epsilons (Theorems~\ref{thm-stratified-proj} and~\ref{thm-R2-variance}), 
a literal analogue of the coefficient of determination from regression statistics.

Applied to the standard alternatives to the Shapley value, the framework yields a clean three-way classification: 
at $n = 4$, the Banzhaf value is nearly orthogonal to the egalitarian Shapley axis ($R^2 \approx 1\%$), 
the equal-surplus-division value is moderately aligned ($R^2 \approx 38\%$), 
and the solidarity value is almost entirely aligned ($R^2 \approx 99.6\%$). 
Asymptotically, this contrast sharpens further. 
The deviations of $\ESD$ and $\So$ from the egalitarian Shapley line are concentrated in directions whose weight in the stratified picture vanishes as $n \to \infty$, 
and consequently $R^2(\ESD) \to 1$ and $R^2(\So) \to 1$. 
The deviation of $\Bz$, by contrast, is governed by the structural identity $\delta_a(\Bz) = -\eps_a(\Bz)/n$ between the efficiency defect and the egalitarian-Shapley deviation at every stratum, 
which forces the squared norms of the efficient and uniform components to match asymptotically, 
hence $R^2(\Bz) \to 1/2$. 
The three-way classification is therefore not a numerical coincidence at $n = 4$ but a structural statement: $\ESD$, $\So$, and $\Bz$ each occupy a distinct location in $\mL^{\mathrm{S}}$, distinguished by the geometric type of their deviation from $\Sh$ (Section~\ref{subsec-classification-structural}).

Several directions remain open.

\paragraph{Other linear targets.} 
The framework applies to any linear value map. 
While in this paper we have computed $\eps^*$, $R^2$ and the residual $\Phi_\perp$ for the Banzhaf, equal-surplus-division, 
and solidarity values, several other natural solutions remain to be analyzed: the weighted Shapley values of 
Kalai and Samet~\cite{kalaisamet1987}, the Owen value~\cite{owen1977} for games with coalition structures, 
and the discounted Shapley values would all extend the picture and likely produce further qualitative findings about the geometry of $\mL$.

\paragraph{Sharper asymptotic results via the spectral decomposition.} 
Theorems~\ref{thm-asym-ESD} and~\ref{thm-asym-Bz} establish the limiting 
behavior of $R^2$ for $\ESD$ and $\Bz$ via direct norm calculations. 
The spectral decomposition of Section~\ref{sec-spectral} provides a unified 
alternative: each result reduces to computing the limit of a weighted mean 
or weighted variance of stratified epsilons, often more transparently 
(see Equation~\eqref{eq-So-asymp} and Corollary~\ref{cor-Bz-asym-stratified}). 
Sharpening the convergence rates is a natural next step:
 in particular,  one can characterize the precise asymptotic of $1 - R^2(\So)$ as $n \to \infty$ 
in terms of the dispersion of $\eps_a(\So)$ across strata.

\paragraph{Higher-dimensional projection.} 
Section~\ref{sec-multiparam} treats projection onto $\mathrm{span}(\Sh, \ED, \Bz)$. 
Larger canonical bases of $\mL$ 
(for instance, including the equal-surplus-division and solidarity values)
 give finer resolution of any target in terms of its decomposition into known solutions.

\paragraph{Beyond linearity.} 
Theorem~\ref{thm-basis-indep} fails for nonlinear value maps. 
A natural question is whether some weaker invariance survives in particular subclasses of nonlinear maps, 
or whether a different inner product structure can be defined that handles nonlinear targets canonically.

\paragraph{Axiomatic interpretation.} 
Given that $\eps^*(\Psi)$ has a closed form, it is natural to ask for an axiomatic characterization of the value $F^{\eps^*(\Psi)}$ for each fixed target $\Psi$; 
in effect, a ``best-fit'' axiomatic characterization that combines the egalitarian-null-player axiom with a minimum-distance condition. 
We leave this for future work.

\paragraph{Statistical analogy.} 
Theorem~\ref{thm-R2-variance} makes the regression analogy literal: 
$R^2(\Psi)$ is the squared ratio of the weighted mean of the stratified 
epsilons to their weighted root-mean-square, with explicit weights 
$w_a = \binom{n}{a}(1/a - 1/n)/D_n$ over coalition sizes. 
Techniques from  regression analysis (such as partial $R^2$, hierarchical decomposition, model 
selection criteria, and leave-one-stratum-out diagnostics) thus admit direct 
analogues for value maps, with ``regressors'' being canonical efficient 
symmetric linear solutions, the ``response'' an arbitrary efficient  symmetric linear $\Psi$, 
and the underlying probability space being  $(\{1, \ldots, n-1\}, w)$.
Pursuing this analogy systematically is left  for future work.

\paragraph{Distributive interpretation.} 
The Shapley line $\{\Feps \mid \eps \in [0,1]\}$ admits a natural reading as a one-parameter family of distributive philosophies, 
interpolating between strict marginalism ($\eps = 0$, the Shapley value: 
each agent receives a payoff determined by their marginal contributions) and strict egalitarianism ($\eps = 1$, 
the equal-division solution: each agent receives the same share regardless of contribution). 
The $\eps$-Dummy axiom can be read in the same spirit, as the axiomatic skeleton of a guaranteed minimum: 
a player who contributes nothing nonetheless receives a positive fraction $\eps$ of the average grand-coalition worth.
Within this reading, our results acquire a corresponding interpretation. 
The near-orthogonality of the Banzhaf value to the marginalism--egalitarianism axis, both at $n = 4$ and asymptotically 
(in the sense that $R^2(\Bz) \to 1/2$), suggests that Banzhaf-style power-index reasoning is not a third position along the marginalism-to-egalitarianism continuum, but rather expresses a logically independent concern; 
namely, the measurement of influence or decision-making power, rather than the allocation of value. 
 The near-alignment of the solidarity value, by contrast, suggests that the Nowak--Radzik construction, 
 despite its apparently distinct motivation (averaging marginal contributions within coalitions), 
 reduces in our geometry to a member of the egalitarian Shapley family. 
  We make no claim that this geometric reduction settles the philosophical comparison between these solutions; 
  we observe only that the inner-product structure on $\mL$ provides a quantitative vocabulary 
  for distinctions that have until now been expressed only informally.
The spectral decomposition adds a further interpretive layer. 
Within  $\mL^{\mathrm{ESL}}$, every value map is a stratified egalitarian
Shapley value with potentially different stratum-specific weightings:
 a solution can be more egalitarian on, say, three-player coalitions than on 
two-player ones, and the framework records exactly this asymmetry. 
Distinct distributive philosophies that ``feel different'' in finite cases (such as 
solidarity versus the egalitarian Shapley family proper) often sit at slightly 
different points in $\R^{n-1}$ but converge to the same diagonal as $n$ grows, 
which is the geometric content of the asymptotic alignment results.

\appendix

\section{Closed Form for the Stratified Epsilons of the Solidarity Value}\label{app-So-formula}

This appendix proves the closed-form formula stated in Example~\ref{ex-So-stratified}:
\begin{equation}\label{eq-app-So-eps-formula}
\eps_a(\So) = a \sum_{s = a+1}^{n} \frac{\binom{n - a - 1}{s - a - 1}}{s \cdot \binom{n-1}{s-1}}, 
\qquad 1 \leq a \leq n - 1.
\end{equation}
This is the direct analogue of Lemma~\ref{lem-So-n4} for general $n$, 
recast in the stratified-flat coordinates of Section~\ref{sec-spectral}. 
Our argument works directly with the value of $\So$ on $u_A$ for $i \notin A$, 
which feeds straight into $\eps_a$ via the identity $\eps_a(\Psi) = n \beta_a(\Psi)$ for any $\Psi \in \mL^{\mathrm{ESL}}$.

\subsection*{The Identity $\eps_a(\Psi) = n \, \beta_a(\Psi)$ for Efficient Symmetric Linear Maps}

\begin{lemma}\label{lem-eps-equals-n-beta}
For every $\Psi \in \mL^{\mathrm{ESL}}$ and every $a \in \{1, \ldots, n-1\}$,
\[
\eps_a(\Psi) = n \, \beta_a(\Psi),
\]
where $\beta_a(\Psi) = \Psi_i(u_A)$ for $i \in N \setminus A$, $|A| = a$.
\end{lemma}

\begin{proof}
By Theorem~\ref{thm-stratified-flat}, $\Psi^{(a)} = (1 - \eps_a) \Sh^{(a)} + \eps_a \ED^{(a)}$. 
Reading the $\beta$-component (i.e., the value at any $i \notin A$),
\[
\beta_a(\Psi) = (1 - \eps_a) \cdot \beta_a(\Sh) + \eps_a \cdot \beta_a(\ED) = (1 - \eps_a) \cdot 0 + \eps_a \cdot \tfrac{1}{n} = \tfrac{\eps_a}{n}.
\]
\end{proof}

By Lemma~\ref{lem-eps-equals-n-beta}, 
computing $\eps_a(\So)$ reduces to computing $\So_i(u_A)$ for $i \notin A$, 
$|A| = a$, which is more direct than computing $(\So - \Sh)_i(u_A)$.

\subsection*{The Defining Sum for $\So_i(u_A)$ When $i \notin A$}

Recall the solidarity value definition (Section~\ref{subsec-solidarity}):
 for $v \in \GN$ and $i \in N$,
\[
\So_i(v) = \sum_{S \ni i} \frac{(n - |S|)! \, (|S| - 1)!}{n!} \cdot A^v(S),
\qquad A^v(S) = \frac{1}{|S|} \sum_{j \in S} (v(S) - v(S \setminus \{j\})).
\]
For the unanimity game $v = u_A$ with $|A| = a$, 
Lemma~\ref{lem-So-unanimity} gives 
\[
A^{u_A}(S) = 
\begin{cases}
a/|S| & \text{if } A \subseteq S, \\
0 & \text{otherwise.}
\end{cases}
\]
Therefore,
\begin{equation}\label{eq-app-So-formula-1}
\So_i(u_A) = \sum_{\substack{S \ni i \\ A \subseteq S}} \frac{(n - |S|)! \, (|S| - 1)!}{n!} \cdot \frac{a}{|S|}.
\end{equation}

\subsection*{Counting the Relevant Subsets When $i \notin A$}

Fix $i \in N \setminus A$. 
The subsets $S$ contributing to~\eqref{eq-app-So-formula-1} satisfy $A \cup \{i\} \subseteq S$. 
For each $s \in \{a+1, \ldots, n\}$, the number of such $S$ with $|S| = s$ equals 
the number of ways to choose the remaining $s - a - 1$ elements of $S$ from the $n - a - 1$ 
elements of $N \setminus (A \cup \{i\})$, namely $\binom{n - a - 1}{s - a - 1}$. 
Substituting into~\eqref{eq-app-So-formula-1},
\begin{equation}\label{eq-app-So-formula-2}
\So_i(u_A) = a \sum_{s = a + 1}^{n} \binom{n - a - 1}{s - a - 1} \cdot \frac{(n - s)! \, (s - 1)!}{n!} \cdot \frac{1}{s}.
\end{equation}

\subsection*{Algebraic Simplification of the Weight}

We rewrite the Shapley-style weight in~\eqref{eq-app-So-formula-2} in a form that makes the dependence on $n$ more transparent. 
Using $(n - s)! \, (s - 1)! / n!  =  1/[n \cdot \binom{n - 1}{s - 1}]$ 
(which follows from the definition of the  binomial   $\binom{n-1}{s-1} = (n-1)!/[(s-1)!(n-s)!]$ and a single cancellation),
\begin{equation}\label{eq-app-So-formula-3}
\So_i(u_A) = \frac{a}{n} \sum_{s = a + 1}^{n} \frac{\binom{n - a - 1}{s - a - 1}}{s \cdot \binom{n - 1}{s - 1}} \qquad (i \notin A).
\end{equation}
The right-hand side depends only on $a$, $s$, and $n$, not on $i$ or the specific choice of $A$, 
confirming the symmetry of $\So$ at the stratum level.

\subsection*{Conclusion}

By Lemma~\ref{lem-eps-equals-n-beta}, $\eps_a(\So) = n \cdot \So_i(u_A)$ for any $i \notin A$, $|A| = a$. 
Combining with~\eqref{eq-app-So-formula-3},
\[
\eps_a(\So) = n \cdot \frac{ a}{n} \sum_{s = a + 1}^{n} \frac{\binom{n -  a - 1}{s - a - 1}}{s \cdot   \binom{n - 1}{s - 1}} = a \sum_{s = a + 1}^{n} \frac{\binom{n - a - 1}{s - a - 1 } }{ s  \cdot \binom{n - 1}{s - 1}},
\]
which is~\eqref{eq-app-So-eps-formula}. \qed

\subsection*{Sanity Check at $n = 4$}

We verify against Lemma~\ref{lem-So-n4}. 
For $a = 3$, the sum reduces to $s = 4$:
\[
\eps_3(\So) = 3 \cdot \frac{\binom{0}{0}}{4 \cdot \binom{3}{3}} = \frac{3}{4}.
\]
For $a = 2$, the sum has terms $s \in \{3, 4\}$:
\[
\eps_2(\So) = 2 \left[\frac{\binom{1}{0}}{3 \cdot \binom{3}{2}}  + \frac{\binom{1}{1}}{ 4 \cdot \binom{ 3}{3}}\right] = 2 \left[\frac{1}{9} + \frac{1}{4}\right] = 2 \cdot \frac{13}{36} = \frac{13}{18}.
\]
For $a = 1$, $s \in \{2, 3, 4\}$:
\[
\eps_1(\So) = 1 \left[\frac{\binom{2}{0}}{2 \cdot \binom{3}{1}} + \frac{\binom{2}{1}}{3 \cdot \binom{3}{2}} + \frac{\binom{2}{2}}{4 \cdot \binom{3}{3}}\right] = \frac{1}{6} + \frac{2}{9} + \frac{1}{4} = \frac{6 + 8 + 9}{36} = \frac{23}{36}.
\]
These reproduce the values $(\eps_1, \eps_2, \eps_3) = (23/36, 13/18, 3/4)$ recorded in Example~\ref{ex-So-stratified}, 
consistent with Lemma~\ref{lem-So-n4} via the stratum-by-stratum identity $(\So - \Sh)^{(a)} = \eps_a (\ED - \Sh)^{(a)}$.

\subsection*{Asymptotic Estimate for Fixed $a$}

We show $\eps_a(\So) \to 1$ for each fixed integer $a$ as $n \to \infty$.

\begin{lemma}\label{lem-So-eps-asymp}
The closed form~\eqref{eq-app-So-eps-formula} can be rewritten as
\begin{equation}\label{eq-app-So-cleaner}
\eps_a(\So) = \frac{a}{\binom{n-1}{a}} \sum_{s=a+1}^{n} \frac{\binom{s-1}{a}}{s} \qquad (1 \leq a \leq n-1).
\end{equation}
For $a = 1$, this gives the exact identity
\begin{equation}\label{eq-app-So-a1-exact}
\eps_1(\So) = 1 - \frac{H_n^{\mathrm{harm}} - 1}{n - 1}, \qquad \text{so} \qquad 1 - \eps_1(\So) = O\!\left(\frac{\log n}{n}\right).
\end{equation}
For every fixed $a \geq 1$, $\eps_a(\So) \to 1$ as $n \to \infty$.
\end{lemma}

\begin{proof}
Form~\eqref{eq-app-So-cleaner} follows from~\eqref{eq-app-So-eps-formula} via the elementary identity
\[
\frac{\binom{n-a-1}{s-a-1}}{\binom{n-1}{s-1}} = \frac{(s-1)(s-2)\cdots(s-a)}{(n-1)(n-2)\cdots(n-a)} = \frac{\binom{s-1}{a}}{\binom{n-1}{a}}.
\]
For $a = 1$, $\binom{s-1}{1} = s-1$ and $\binom{n-1}{1} = n-1$, so 
the proof of~\eqref{eq-app-So-a1-exact} is complete:
\[
\eps_1(\So) = \frac{1}{n-1}\sum_{s=2}^{n}\frac{s-1}{s} = \frac{1}{n-1}\sum_{s=2}^{n}\!\left(1 - \frac{1}{s}\right) = 1 - \frac{H_n^{\mathrm{harm}} - 1}{n-1}.
\]
For general $a$, apply Abel summation to~\eqref{eq-app-So-cleaner} with summands $\binom{s-1}{a}$ and weights $1/s$, 
using the hockey-stick identity 
$\sum_{r=a+1}^{s}\binom{r-1}{a} = \binom{s}{a+1}$ and $\binom{n}{a+1} = \binom{n-1}{a} \cdot n/(a+1)$:
\begin{equation}\label{eq-app-So-abel}
\eps_a(\So) = \frac{a}{a+1} + \frac{a}{\binom{n-1}{a}} \sum_{s=a+1}^{n-1} \frac{\binom{s}{a+1}}{s(s+1)}.
\end{equation}
Setting $x_s := s/n$, the summand satisfies, as $n \to \infty$ uniformly in $s/n \in (0,1)$,
\[
\frac{a}{\binom{n-1}{a}}  \cdot \frac{\binom{s}{a+1}}{s(s+1)} = \frac{a\,(n-1-a)}{a+1}  \cdot \frac{\binom{s}{a+1}/\binom{n-1}{a+1}}{s(s+1)} = \frac{a\, x_s^{a+1}}{(a+1)\, n\, x_s^2}\bigl(1 + O(1/n)\bigr) =  \frac{a\, x_s^{a-1}}{(a+1)\, n}\bigl(1 +  O(1/n)\bigr),
\]
where we used the identities 
 $\binom{s}{a+1}/\binom{n-1}{a+1} = \prod_{j=0}^{a}(s-j)/(n-1-j) = x_s^{a+1}(1 + O(1/n))$ and $(n-1-a)/n = 1 + O(1/n)$. 
 Recognizing the Riemann sum,
\[
\frac{a}{\binom{n-1}{a}} \sum_{s=a+1}^{n-1} \frac{\binom{s}{a+1}}{s(s+1)} \;\longrightarrow\; \int_0^1 \frac{a\, x^{a-1}}{a+1}\, dx \;=\; \frac{1}{a+1}.
\]
Substituting into~\eqref{eq-app-So-abel} gives 
$\eps_a(\So) \to a/(a+1) + 1/(a+1) = 1$.
\end{proof}

Combined with Theorem~\ref{thm-R2-variance}, Lemma~\ref{lem-So-eps-asymp} yields 
Equation~\eqref{eq-So-asymp}: $R^2(\So) \to 1$ as $n \to \infty$.

\begin{remark}
The closed form~\eqref{eq-app-So-eps-formula} expresses $\eps_a(\So)$ as a finite sum of explicit rational terms but does not, to our knowledge, 
simplify further to a closed form involving only standard special functions for general $a$ and $n$. 
Beyond the exact formula~\eqref{eq-app-So-a1-exact} for $a = 1$, the two opposite extreme cases admit clean explicit values: 
$\eps_{n-1}(\So) = (n-1)/n$ (only the $s = n$ term contributes) and $\eps_{n-2}(\So) = (n-2)\bigl[(n-1)^2 + n\bigr] / \bigl[n (n-1)^2\bigr]$.
\end{remark}

\end{document}